\documentclass[10pt,a4paper]{amsart}
\usepackage[margin=20mm]{geometry}
\usepackage[numbers]{natbib}
\usepackage[colorlinks]{hyperref}

\usepackage{amssymb,amsmath}
\usepackage{amsthm}
\usepackage{bbm}
\usepackage{stmaryrd}
\usepackage{hyperref}
\usepackage{color}

\usepackage[all]{xypic}

\numberwithin{equation}{section}

\newtheorem{proposition}{Proposition}[section]

\theoremstyle{definition}
\newtheorem{definition}{Definition}[section]
\newtheorem{remark}{Remark}[section]
\newtheorem{example}{Example}[section]

\newenvironment{warning}[1][Warning.]{\begin{trivlist}
\item[\hskip \labelsep {\bfseries #1}]}{\end{trivlist}}

\newcommand{\Id}{\mathbbmss{1}}

\newcommand{\InHom}{\mbox{$\underline{\Hom}$}}
\newcommand{\InDiff}{\mbox{$\underline{\Diff}$}}

\newcommand{\rmd}{\textnormal{d}}

\newcommand{\rmo}{\textnormal{o}}

\DeclareMathOperator{\Vect}{Vect}

\DeclareMathOperator{\Diff}{Diff}

\DeclareMathOperator{\Hom}{Hom}

\newcommand{\catname}[1]{\textnormal{\texttt{#1}}}

\font\black=cmbx10 \font\sblack=cmbx7 \font\ssblack=cmbx5 \font\blackital=cmmib10  \skewchar\blackital='177
\font\sblackital=cmmib7 \skewchar\sblackital='177 \font\ssblackital=cmmib5 \skewchar\ssblackital='177
\font\sanss=cmss10 \font\ssanss=cmss8 
\font\sssanss=cmss8 scaled 600 \font\blackboard=msbm10 \font\sblackboard=msbm7 \font\ssblackboard=msbm5
\font\caligr=eusm10 \font\scaligr=eusm7 \font\sscaligr=eusm5  \font\fraktur=eufm10
\font\sfraktur=eufm7 \font\ssfraktur=eufm5 
\font\bsymb=cmsy10 scaled\magstep2
\def\all#1{\setbox0=\hbox{\lower1.5pt\hbox{\bsymb
       \char"38}}\setbox1=\hbox{$_{#1}$} \box0\lower2pt\box1\;}
\def\exi#1{\setbox0=\hbox{\lower1.5pt\hbox{\bsymb \char"39}}
       \setbox1=\hbox{$_{#1}$} \box0\lower2pt\box1\;}

\newfam\bifam
\textfont\bifam=\blackital \scriptfont\bifam=\sblackital \scriptscriptfont\bifam=\ssblackital

\newfam\blfam
\textfont\blfam=\black \scriptfont\blfam=\sblack \scriptscriptfont\blfam=\ssblack

\newfam\bbfam
\textfont\bbfam=\blackboard \scriptfont\bbfam=\sblackboard \scriptscriptfont\bbfam=\ssblackboard

\newfam\ssfam
\textfont\ssfam=\sanss \scriptfont\ssfam=\ssanss \scriptscriptfont\ssfam=\sssanss
\def\sss#1{{\fam\ssfam\relax#1}}

\newfam\clfam
\textfont\clfam=\caligr \scriptfont\clfam=\scaligr \scriptscriptfont\clfam=\sscaligr

\newfam\frfam
\textfont\frfam=\fraktur \scriptfont\frfam=\sfraktur \scriptscriptfont\frfam=\ssfraktur

\def\hpb#1{\setbox0=\hbox{${#1}$}
    \copy0 \kern-\wd0 \kern.2pt \box0}
\def\vpb#1{\setbox0=\hbox{${#1}$}
    \copy0 \kern-\wd0 \raise.08pt \box0}

\def\pmb#1{\setbox0\hbox{${#1}$} \copy0 \kern-\wd0 \kern.2pt \box0}
\def\pmbb#1{\setbox0\hbox{${#1}$} \copy0 \kern-\wd0
      \kern.2pt \copy0 \kern-\wd0 \kern.2pt \box0}
\def\pmbbb#1{\setbox0\hbox{${#1}$} \copy0 \kern-\wd0
      \kern.2pt \copy0 \kern-\wd0 \kern.2pt
    \copy0 \kern-\wd0 \kern.2pt \box0}
\def\pmxb#1{\setbox0\hbox{${#1}$} \copy0 \kern-\wd0
      \kern.2pt \copy0 \kern-\wd0 \kern.2pt
      \copy0 \kern-\wd0 \kern.2pt \copy0 \kern-\wd0 \kern.2pt \box0}
\def\pmxbb#1{\setbox0\hbox{${#1}$} \copy0 \kern-\wd0 \kern.2pt
      \copy0 \kern-\wd0 \kern.2pt
      \copy0 \kern-\wd0 \kern.2pt \copy0 \kern-\wd0 \kern.2pt
      \copy0 \kern-\wd0 \kern.2pt \box0}

\mathchardef\za="710B  
\mathchardef\zb="710C  
\mathchardef\zg="710D  
\mathchardef\zd="710E  
\mathchardef\zve="710F 
\mathchardef\zz="7110  
\mathchardef\zh="7111  
\mathchardef\zvy="7112 
\mathchardef\zi="7113  
\mathchardef\zk="7114  
\mathchardef\zl="7115  
\mathchardef\zm="7116  
\mathchardef\zn="7117  
\mathchardef\zx="7118  
\mathchardef\zp="7119  
\mathchardef\zr="711A  
\mathchardef\zs="711B  
\mathchardef\zt="711C  
\mathchardef\zu="711D  
\mathchardef\zvf="711E 
\mathchardef\zq="711F  
\mathchardef\zc="7120  
\mathchardef\zw="7121  
\mathchardef\ze="7122  
\mathchardef\zy="7123  
\mathchardef\zf="7124  
\mathchardef\zvr="7125 
\mathchardef\zvs="7126 
\mathchardef\zf="7127  
\mathchardef\zG="7000  
\mathchardef\zD="7001  
\mathchardef\zY="7002  
\mathchardef\zL="7003  
\mathchardef\zX="7004  
\mathchardef\zP="7005  
\mathchardef\zS="7006  
\mathchardef\zU="7007  
\mathchardef\zF="7008  
\mathchardef\zW="700A  
\mathchardef\zC="7009  

\newcommand{\be}{\begin{equation}}
\newcommand{\ee}{\end{equation}}

\newcommand{\bea}{\begin{eqnarray}}
\newcommand{\eea}{\end{eqnarray}}
\def\*{{\textstyle *}}
\newcommand{\R}{{\mathbb R}}

\newcommand{\s}{{\textstyle *}}

\def\Hom{\sss{Hom}}

\def\sT{{\sss T}}

\def\st{{\sss t}}

\def\s*{{\scriptstyle *}}

\newcommand{\beas}{\begin{eqnarray*}}
\newcommand{\eeas}{\end{eqnarray*}}

\title{On a geometric framework for Lagrangian supermechanics}

   \author{Andrew James Bruce}

  \address{IMPAN, ul. Sniadeckich 8, 00--656 Warsaw, POLAND.} \email{andrewjamesbruce@googlemail.com}
 \author{Katarzyna  Grabowska} 

   \address{Faculty of Physics, University of Warsaw, Pasteura 5, 02-093 Warszawa, POLAND.} \email{konieczn@fuw.edu.pl}

 \author{Giovanni Moreno}

   \address{IMPAN, ul. Sniadeckich 8, 00--656 Warsaw, POLAND.} \email{gmoreno@impan.pl}

\date{\today}
 
\begin{document}

\begin{abstract}
We re--examine classical mechanics with both commuting and anticommuting  degrees of freedom. We do this by defining the phase dynamics of a general Lagrangian system  as an implicit differential equation in  the spirit of Tulczyjew. Rather than parametrising  our basic degrees of freedom   by a specified Grassmann algebra,    we use arbitrary supermanifolds by following the categorical approach to supermanifolds.\par
\smallskip\noindent
{\bf Keywords:} 
supermechanics; supermanifolds; Lagrangian systems; Tulczyjew triples; curves.\par
\smallskip\noindent
{\bf MSC 2010:}  58A50;  58C50;  70H33;  70H99; 70G45.
\end{abstract}


 \maketitle

\setcounter{tocdepth}{2}
 \tableofcontents

\section{Introduction}\label{sec:int}
The classical treatment of field theories with both bosonic and fermionic degrees of freedom was born from the need to develop phenomenologically reasonable quantum theories. Since a true classical treatment of fermions is arguably impossible,   by \emph{classical} we will really mean \emph{quasi--classical} in the sense that   both commuting and anticommuting degrees of freedom are allowed  in the theory.  Just to stress the importance  of quasi--classical theories, recall that without them one would have to   write down directly quantum theories   involving both fermions and bosons---a  presently impossible task.  Moreover, it is common in semiclassical treatments of quantum field theory to set the fermionic degrees of freedom to zero once the supersymemtric theory has been constructed: for example Witten \cite{Witten:1981} does exactly this when discussing instantons in a model of supersymmetric quantum mechanics. However, we know from the work of Akulov \& Duplij \cite{Akulov:1997} that this is not always the correct thing to do. Thus, having some handle on quasi--classical backgrounds is necessary in general. Another example is provided by gauge theories, which via  perturbative methods  lead to  anticommuting fields --- the  well-known Faddeev--Popov ghosts --- even if no fermions are present in the theory.\par
In this paper we re--examine  \emph{supermechanics} from a supergeometric point of view, which we understand  as \emph{mechanics on supermanifolds}: this is quite independent of the notion of supersymmetry. In this perspective,  dynamical variables   depend on time only, though defining curves on a supermanifold needs  particular care. We will employ Grothendieck's functor of points, and the corresponding description of mapping spaces between supermanifolds. In particular,  a curve on a supermanifold will be an $S$--curve \cite{Bruce:2014}, which we can think of as a family of curves parametrised by an arbitrary supermanifold. In spite of its  abstractness, such a categorical approach brings some rigour  to the more informal handling of odd fields as found throughout the physics literature, see for example \cite{Freed:1999}. \par
We take the point of view that the only physically meaning parameter for a curve on a supermanifold is time. However,   a curve na\"{\i}vely defined  as a (categorical) morphism between supermanifolds $\mathbb{R} \longrightarrow M$ clearly `misses' the odd (i.e.,  anticommuting) part of the supermanifold $M$, since the categorical maps between supermanifolds must preserve the Grassmann parity. Thus  only   the topological points of the underlying manifold can be reached by such a narrow notion of a curve. In order to overcome this obstacle, without having  to introduce some notion of `super--time', we are forced to consider the trajectories of a quasi--classical system as being $S$--curves (or $S$--paths).\par
The necessity of parametrising curves by external odd constants was long known in the  physical literature  \cite{Berezin:1977, Casalbuoni:1976}, though,    of course, phrased differently as requiring both the fermionic and bosonic degrees of freedom to take their values in some chosen Grassmann algebra (for certain models \cite{Junker:1996,Junker:1994, Junker:1995,Manton:1999} this may even be chosen to consist of two generators only). A drawback of   these earlier works is the lack of a proper discussion about the  separation of the true dynamical degrees of freedom from the external parametrisation. 
 A first clue on the functorial behaviour of curves in supermechanics can be found in the   work of Heumann \& Manton \cite{Heumann:2000} who discuss, in the context of a particular model of supersymmetric mechanics (which is related to Witten's model \cite{Witten:1981}), general Grassmann algebra valued equations of motion.  Indeed, they show that, for a large class of potentials in their model, exact solutions to the equations of motion can be found without specifying exactly the Grassmann algebra that everything takes values in.  From the $S$--points perspective,  this is not surprising: the Grassmann algebra used to parametrise the dynamics should not play a vital r\^ole, other than ensuring that the encountered expressions are of the right Grassmann parity.   Even if we will use all supermanifolds as our parametrisations, and not just Grassmann algebras, we stress   that (finite dimensional) Grasssmann algebras are enough to parametrise a supermaifold via the functor of points, according to   Schwarz \& Voronov \cite{Shvarts:1984,Voronov:1984}.\par
The notion of an $S$--curve and how to use them to define higher--order tangent bundles was the subject of \cite{Bruce:2014}.  Armed with a robust and appropriate formalism, the next step we are going to take is to apply  Tulczyjew's geometric formalism  \cite{Tulczyjew:1977} to  supermechanics. We will stick with the Lagrangian side of the \emph{Tulczyjew triple}, and geometrically define the phase dynamics (the Euler--Lagrange equations are  a sort of `shadow' of the phase dynamics) and solutions thereof. In particular, we work with autonomous Lagrangians. Furthermore, we will restrict our attention to Grassmann even Lagrangians (see Remark \ref{rem4Janusz}).\par
It is worth pointing out that Tulczyjew's geometric formulation of classical mechanics   deals with both nondegenerate and degenerate Lagrangians in a unified way. Recall that, in general, nondegnerate Lagrangians with odd degrees of freedom lead to quantum theories with negative norm states --- the `bad ghosts'. For the case of gauge--fixed Lagrangians with Faddeev--Popov  ghosts (`good ghosts'), which are constructed to be nondegenerate, the ghosts are unphysical and are removed from the theory using the BRST operator. However, if we want to consider physical fermions then we have to avoid quantum states with negative norm in another way. Typically, physically interesting Lagrangians with odd degrees of freedom are linear in first derivatives of the odd degrees of freedom and thus are degenerate (or singular in the language of Dirac and Bergmann). However, such degenerate Lagrangians lead to quantum theories that are free of negative norm states.  Hence, a common geometric framework for  both singular and regular Lagrangian is of paramount importance for supermechanics.\par
Our intention is to show how to define \emph{geometrically} and \emph{rigorously} the phase dynamics of a supermechanical system  as a sub-supermanifold of $\sT \sT^{*}M$. Generally  we will be considering implicit phase dynamics, i.e., the dynamics is not given by the `image' of an even vector field on $\sT^{*}M$. We then define solutions to the phase dynamics in terms of $S$--curves on $\sT^{*}M$ and their tangent prolongations. We view finding solutions for particular models as a secondary question, which will require ad hoc  methods (see \cite{Heumann:2000,Junker:1996,Manton:1999} for examples).  The primary aim of this paper is to find the proper (and minimalistic) geometric framework for supermechanics.\par
Extending the  methods of geometric mechanics to supermechanics is of course not a new idea. The first work in this direction that we are aware of is Ibort \& Mar\'{\i}n-Solano \cite{Ibort:1993}. We must also mention Cari\~{n}ena \& Figueroa, who  in a series of papers that include \cite{Carinena:1994,Carinena:1997,Carinena:2003} also investigated extending various ideas from geometric mechanics to supermechanics. However, the approaches of these works are different than the one explored here, and they do not properly explain the dynamics in terms of a suitable generalisation of a curve on a supermanifold. We just  follow the classical ideas of curves and their tangent prolongations  very closely, and  we manage to avoid the unnatural trick of enlarging   (the total space of) the tangent bundle of a supermanifold. By dealing with curves that depend on time and not `super-time' we have well--defined tangent (and higher tangent) prolongations. \par
 Supermechanics cannot be expected to describe  directly our macroscopic world in some approximate sense as standard classical mechanics does. In particular the classical limit is understood as the limit for which Planck's constant tends to zero while at the same time allowing very large quantum numbers. The problem is that for anticommuting degrees of freedom that obey the spin--statistic theorem, as all physical fermions must do, the Pauli exclusion principle forbids arbitrary large quantum numbers. This means that even in the quasi--classical limit where we allow anticommuting objects, we cannot really understand the systems in the same light as classical mechanics. In further generality, there is no reason to expect supermanifolds to play such a fundamental and direct r\^{o}le in understanding nature  as manifolds do. Nonetheless, supermanifolds are vital in bridging the gap between the truly classical world and the quantum one.  
\subsubsection*{Main results} In this paper we
\begin{enumerate}
\item define the phase dynamics of a supermechanical (Lagrangian) system in the spirit of Tulczyjew.
\item define the solutions of the phase dynamics in terms of $S$--curves (or $S$--paths).
\item present infinitesimal (super)symmetries of the phase dynamics as vector fields on the phase space of the supermechanical system.
\end{enumerate}
Needless to say, the Euler--Lagrange equations for a supermechanical Lagrangian can of course be derived using a formal variation with respect to the even and odd degrees of freedom in the theory. The main novelty here is the geometric method we apply via $S$--curves: in doing so we avoid having to deal with the calculus of variations.  Actually, there is no clear and general consensus about the notion of an ordinary differential equation and their solutions on a supermanifold. Indeed, different ideas come from different understandings of the classical situation. For the present purposes, a  differential equation (explicit or implicit) is a substructure of the tangent bundle of a supermanifold,   understood in terms of the functor of points, and a solution is defined in terms of $S$--curves and their tangent lift. In conclusion, one just needs to employ the functor of points and internal Homs and to  follow  closely the classical geometric notions.
\subsubsection*{Arrangement} Although we assume the reader has some familiarity with the theory of supermanifolds, in Section \ref{sec:pre} we review the basic theory  as needed in the remaining sections of this paper. In particular we sketch the notion of the functor of points, generalised supermanifolds as functors and the generalised supermanifold of maps between given supermanifolds. In this section we also briefly recall the notion of $S$--curves as given in \cite{Bruce:2014}. In section \ref{sec:Lag} we define first--order mechanical Lagrangians on a supermanifold and show how to geometrically define their phase dynamics, in close analogy with  Tulczyjew's framework for classical mechanics. In Section \ref{sec:examples} we present some simple supermechanical systems in order to illustrate our main constructions. We end this paper with some brief comments in Section \ref{sec:con}. We include two short appendice: Appendix \ref{AppTM} on the tangent and cotangent bundles of supermanifolds, and Appendix \ref{AppDVBs} on the canonical diffeomorphism $\alpha: \sT \sT^{*}M \rightarrow \sT^{*}\sT M$.
\section{Preliminaries on supermanifolds}\label{sec:pre}
\subsection{Supermanifolds and their morphsims}
 We will follow the `Russian school' and denote by $\catname{SM}$ the category of real finite--dimensional supermanifolds understood as locally superringed spaces \cite{Carmeli:2011,MR1701597,Manin:1997,Varadarajan:2004}. A \emph{supermanifold}   $M$ is understood to be defined by its structure sheaf $(|M|,\mathcal{O}_{M})$, where $|M|$ is the \emph{reduced manifold} (or \emph{body}) of $M$. By an \emph{open superdomain} $U$ of $M$, we mean an open domain in $  |M|$, together with the corresponding restriction of the structure sheaf. Sections of the structure sheaf will be called \emph{functions} on $M$ and the set of all functions will be denoted by $C^{\infty}(M)$.\par
A \emph{morphism} between supermanifolds $\psi : M \longrightarrow N$ is a pair of morphisms $(|\psi|, \psi^{*})$ where$ |\psi| : |M|\longrightarrow  |N|$ is a continuous map and $\psi^{*}: \mathcal{O}_{N} \longrightarrow \mathcal{O}_{M}$ is a sheaf morphism above $|\psi|$. The set of morphisms between a pair of supermanifolds will be denoted by $\Hom_{\catname{SM}}(M,N) := \Hom(M,N)$. Note that these categorical morphisms necessarily preserve the $\mathbb{Z}_{2}$-grading.  For simplicity we will assume that all the supermanifolds $M$ we work with are connected, that is $|M|$ will be connected.\par
Recall that  as supermanifolds are locally diffeomorphic to  $\mathbb{R}^{p|q}$ (for some $p$ and $q\in \mathbb{N}$), we can work with \emph{local coordinate patches} on supermanifolds. We will typically group even (denoted by $x^\mu$) and odd coordinates (denoted by $\theta^\alpha$), together and present them as $(x^{a})$ and denote the Grassmann parity by $\widetilde{x^{a}} := \widetilde{a} \in \{0,1\}$. The coordinates of a supermanifold are commuting in the graded sense. What makes supermanifolds so workable is that morphisms  and changes of local coordinates can be written in exactly the same way as one would in the category of smooth manifolds. Note that in general one allows changes of coordinates that mix even and odd coordinates provided they preserve the Grassmann parity. In particular, although we have the famous Gaw\c{e}dzki--Batchelor theorem \cite{Batchelor:1979,Gawedzki:1977} --- every (real) supermanifold is noncanonically isomorphic to a bundle of Grassmann algebras --- supermanifolds do not in general have a fibre bundle structure over a manifold. That is, we will not simply be working with  `Gaw\c{e}dzki--Batchelor models', but rather with the full category of supermanifolds. Also, note that we have a bijection between $C^{\infty}(M)$ and $\Hom(M,\mathbb{R}^{1|1})$ simply given by $f \longmapsto (t \circ f, \tau \circ f)$, where we have chosen global coordinates $(t, \tau)$ on $\mathbb{R}^{1|1}$. \par
Tensor and tensor--like objects, for example vector fields, also naturally carry Grassmann parity and we   use the  notation  `tilde', as we have done for coordinate functions, to denote their parity. By an \emph{even} or \emph{odd} object we mean an object that carries Grassmann parity of zero or one respectively. \par

 It is important to note that supermanifolds are not truly set--theoretical objects and that they represent a class of (mildly) noncommutative geometries. The only `true points' of a supermanifold $M$ are the points of the reduced manifold $|M|$, usually referred to as the \emph{topological points} of $M$. One should think of the points of a supermanifold as the topological points of the reduced manifold surrounded by---so to speak---anticommutative `fluff' that cannot be non--trivially localised. Thus, one will have to take care with generalising point--wise notions from the category of manifolds to the category of supermanifolds. The functor of points, which we will discuss next provides a way of recovering some classical intuition.

\subsection{The functor of points} We will need a more categorical approach to supermanifolds and in particular the notion of the \emph{functor of points} (discussed here) and the \emph{internal Homs} (discussed in \ref{subIntHoms} below). The \emph{$S$--points} of a supermanifold $M$ are elements in the set $\Hom(S,M)$, where $S$ is some arbitrary supermanifold. That is, one can view a supermanifold as a functor
\begin{eqnarray}
M  : \catname{SM}^{\rmo}  &\longrightarrow&  \catname{Set}\, ,\label{eqSTAR}\\
S &\longmapsto &\Hom(S,M) := M(S)\, ,\nonumber
\end{eqnarray}
which is an example of the Yoneda embedding (see, e.g., \cite{MacLane1998}, III.2). Via Yoneda's lemma, we can identify a supermanifold  $M$ with the functor \eqref{eqSTAR}, in such a way that   morphisms between supermanifolds correspond to natural transformations between the corresponding functors. Such natural transformations amount to maps between the respective sets of $S$--points. Informally, one can think about the $S$--points of $M$ as being parametrised by \emph{all} supermanifolds.

\begin{remark}\label{rem21}There is a nice particle physics analogy here with the functor of points, which we think is originally due to Ravi Vakil. Imagine that we   want to discover all the properties of an unknown particle, and to this end we fire at it \emph{all} possible test particles with a wide range of energies.  By observing the results of the  interactions   with \emph{all} the test particles over \emph{all} energy scales, we can deduce \emph{all} the properties of the unknown particle and thus characterise it completely. The functor of points allows us to `probe' a given supermanifold by observing how, as a functor, it `interacts' with all `test supermanifolds' (the latter being often referred to as \emph{parametrisations}). However life is a little simpler than that: due to the  work of Schwarz and Voronov \cite{Shvarts:1984,Voronov:1984} it is known that it is actually sufficient to probe the supermanifold under study  with  parametrisations of the form $\mathbb{R}^{0|q}$ $(q \geq 1)$ only, i.e., the  simpler Grassmann algebras. This is  not unlike a typical particle physics   experiment, where    we can gain enough information to understand the standard model, even without  colliding all possible particles together nor   exploring all energies.
\end{remark}

\begin{remark}\label{rem22}
As the functor of points involves maps between finite--dimensional supermanifolds, one can consider $S$--points locally via coordinates.  In particular, if we employ some coordinate system 
\begin{equation}\label{eqLocCoord}
(x^{a}) = (x^{\mu},\theta^{\alpha})
\end{equation}
 on a local superdomian $U$ of $M$, then the $S$--points are   specified by systems of functions $(x_{S}^{\mu}, \theta_{S}^{\alpha})$, where $x^{i}_{S}$ is a collection of even functions on $S$ and similarly $\theta_{S}^{\alpha}$ is a collection of odd function on $S$. As the
supermanifold $S$ is chosen arbitrarily dependence on the local coordinates of S will not explicitly be presented.
\end{remark}
Given a morphism 
\begin{equation}\label{eqChangParam}
\psi \in \Hom(P, S)
\end{equation}
 between supermanifolds $P$ and $S$ we have an induced map of sets
\begin{eqnarray}
\Psi   :  M(S)& \longrightarrow& M(P)\, ,\label{eqSTARSTAR}\\
m & \longmapsto& m \circ \psi\, ,\nonumber
\end{eqnarray}
 for all $m \in M(S)$. Thus we speak of the \emph{functor of points} \eqref{eqSTAR}--\eqref{eqSTARSTAR}. The morphism \eqref{eqChangParam} is usually referred to as a \emph{change of parametrisation}.

\begin{warning}
On occasion we will use an abuse of set--theoretical notation and terminology. As we are   not  really dealing with sets, our language and notation maybe somewhat inappropriate. Thus, we will always consider the meaning  of set--theoretical notions in terms of $S$--points and not topological points. 
\end{warning}

\subsection{Internal Homs and generalised supermanifolds}\label{subIntHoms}
 A \emph{generalised supermanifold} is an object in the functor category $\widehat{\catname{SM}} := \catname{Fun}(\catname{SM}^{\rmo}, \catname{Set})$, whose objects are functors from $\catname{SM}^{\rmo}$ to the category $\catname{Set}$ and whose morphisms are natural transformations \cite{Shvarts:1984}. Note that this functor category contains $\catname{SM}$ as a full subcategory via the Yoneda embedding \eqref{eqSTAR}. One passes from the category of finite--dimensional supermanifolds to the larger category of generalised supermanifolds in order to understand the \emph{internal Homs} objects, henceforth denoted by $\InHom$. In particular there always exists a generalised supermanifold such that the so--called \emph{adjunction formula} holds
\begin{equation*}
\InHom(M,N)(\bullet) := \Hom(\bullet \times M,N) \in  \widehat{\catname{SM}}\, .
\end{equation*}
Heuristically, one should think of enriching the morphisms between supermanifolds to now have the structure of a supermanifold, however to understand this one passes to a larger category. In essence we will use the above to define what we mean by a \emph{mapping supermanifold} and will probe it using the functor of points. We will refer to `elements' of a mapping supermanfold as \emph{maps} reserving morphisms for the categorical morphisms of supermanifolds.\par
A generalised supermanifold is \emph{representable} if it is naturally isomorphic to a supermanifold in the image of the Yoneda embedding \eqref{eqSTAR}.

\begin{example}
 It is easy to see that $\Hom(\{\emptyset\},M) = |M|$, while $\InHom(\{\emptyset\},M) = M$.
\end{example} 

\begin{example}\label{exAntTangBund}
  Another well--known example of a representable generalised supermanifold is the antitangent bundle $\InHom(\mathbb{R}^{0|1},M) = \Pi \sT M$. In fact it is well--known that the generalised supermanifolds $\InHom(\mathbb{R}^{0|p}, M)$  are representable for all $p \in \mathbb{N}$.
\end{example}

The composition of maps between supermanifolds can be thought of in terms of a natural transformation
\begin{equation}
\underline{\circ} : \InHom(M,N) \times \InHom(N,L) \longrightarrow \InHom(M,L)\, ,\label{eq5Star}
\end{equation}
defined, for any $S \in \catname{SM}$,  by
\begin{eqnarray}
 \Hom(S\times M, N) \times \Hom(S\times N,L)& \longrightarrow& \Hom(S\times M, L)\, ,\label{eq5Star5Star}\\
  (\Phi_{S}, \Psi_{S})& \longmapsto &(\Psi \underline{\circ} \Phi)_{S}:= \Psi_{S} \circ (\Id_{S}, \Phi_{S})\circ (\Delta, \Id_{M})\, ,\nonumber
\end{eqnarray}
and then letting $S$ vary   over all supermanifolds. Here $\Delta:S\longrightarrow S\times S$ is the diagonal of $S$ and $\Id_S:S\longrightarrow S$ is the identity.
\subsection{Curves on supermanifolds}  The notion of a \emph{curve} on a supermanifold requires some delicate handling and we are forced to adopt a very categorical framework. This may at first seem somewhat disconnected from the classical notion of a curve on a manifold, but we will comment on this shortly.

\begin{definition}
A  \emph{curve on a supermanifold} $M$ (\emph{$S$--curve}, for brevity)  is a functor $\gamma \in \InHom(\mathbb{R}, M)$.  
\end{definition}

To make proper sense of this we `probe' the generalised supermanifold of curves using the functor of points \eqref{eqSTAR}--\eqref{eqSTARSTAR}. Given any supermanifold $S$ an $S$--curve $\gamma$ assigns a morphism
\begin{equation*}
\gamma_{S} \in \InHom(\mathbb{R}, M)(S)= \Hom(S \times \mathbb{R}, M)\, .
\end{equation*}
To any $\psi \in \Hom(P,S)$  we have an induced a morphism $\overline{\psi}: \InHom(\mathbb{R}, M)(S) \longrightarrow \InHom(\mathbb{R}, M)(P)$ given by
\begin{equation*}
\gamma_{S} \longmapsto \gamma_{P} := \gamma_{S} \circ (\psi, \Id_{\mathbb{R}})\, .
\end{equation*}
That is, we should consider an $S$--curve $\gamma$ as a family $\{ \gamma|_t \}_{t\in\mathbb{R}}$ of functors from the (opposite) category of supermanifolds to sets that is parametrised by the real line, viz.
\begin{equation*}
\gamma|_t :S\in \catname{SM} \longmapsto \gamma_S|_t \in \Hom(S,M)\, .
\end{equation*}
Alternatively, by exchanging the r\^{o}les of $S$ and $\mathbb{R}$, we may think of $\gamma$ as a family  $\{ \gamma|_S \}_{S\in\catname{SM}}$ of maps
\begin{equation}\label{eqGammaS}
\gamma|_S :t\in \mathbb{R} \longmapsto \gamma|_{S\times\{t\}} \in \Hom(S,M)\, .
\end{equation}
That is, since $\gamma|_{S\times\{t\}}$ depends on $S$, $\gamma|_S:S\times\mathbb{R}\longrightarrow M$ is a family of curves in $M$ parametrised by  $S$.
\begin{warning}
We will refer to \emph{both} $\gamma$ and $\gamma_{S}$, where $S$ is an arbitrary supermanifold, as $S$--curves. The context should make the meaning clear.
\end{warning}

In practice we may also consider $S$--paths by replacing $\mathbb{R}$ with an  open interval. For a chosen, but arbitrary supermanifold $S$,   an $S$--curve is a family \eqref{eqGammaS} of   curves   parametrised by $S$: this external parametrisation provides precisely the `extra oddness' that is needed so that the curve does not miss the odd directions of $M$. The image set of an $S$--curve  $\gamma_{S}(\mathbb{R}) \subset M(S)$ is a collection of $S$--points of $M$. This is close to the classical notion of a curve which traces out the topological points of a manifold. In the supercase, $S$--curves trace out $S$--points of a supermanifold.  In particular, we can always consider $\gamma|_{t_{0}} : \InHom(\mathbb{R}, M) \longrightarrow M$ as a natural transformation for any  $t_{0} \in \mathbb{R}$ and making note that $M = \InHom(\{ \emptyset\},M)$. The statement that an $S$--curve passes through an $S$--point $m \in M(S)$ means $\gamma_{S}(0)=m$.\par

As $S$--curves, once a supermanifold $S$ has been chosen,  are standard categorical morphisms between supermanifolds, we can describe them locally in terms of coordinates. Let us consider some local coordinate system $(x^{a}) =  (x^{\mu}, \theta^{\alpha})$ on a local superdomain  $U$ of  $M$. We also employ   the global standard coordinate system $(t)$ on $\mathbb{R}$, and bear in mind Remark \ref{rem22}.  Then any $S$--curve can be written locally as
$$(x^{a}\circ \gamma_{S})(t) = (x^{\mu}_{S}(t), \theta^{\alpha}_{S}(t)),$$
which is a system of even and odd functions on $S \times \mathbb{R}$.   We stress that typographically adding a subscript `$S$' to our local expressions is more than just a formal trick, but has genuine meaning within the context of the functor of points (see Example \ref{exExplicitDE} later on, where we examine  the functorial behaviour of the $S$--curves satisfying a certain equation).\par
Once a parametrising supermanifold $S \in \catname{SM}$ has been chosen, $S$--curves are  formally reminiscent of \emph{homotopies} of standard classical curves. In particular we use the extra `oddness' provided by $S$ to `push' or `deform' the curve off the reduced manifold $|M|$ and so detect the odd directions of the supermanifold $M$.\par
To see this, it suffices to consider $S = \mathbb{R}^{0|p}$ (see Remark \ref{rem21}), and  suppose that  $M$ is  equipped with local coordinates $(x^{\mu}, \theta^{\alpha})$. Then any $S$--curve $\gamma$  in the `even directions' locally looks like
\begin{equation}
x^{\mu} \circ \gamma_{\mathbb{R}^{0|p}}(t) =  x^{\mu}(t) + \sum_{n\ \textnormal{even}} \frac{1}{n!} \zx^{i_{1}}\zx^{i_{2}} \cdots \zx^{i_{n}}x^{\mu}_{i_{n}i_{n-1} \cdots  i_{1}}(t)\, .
\end{equation}
 This means that the components of $\gamma$ in the  `even directions' are described by a collection of standard curves on $|M|$. These curves   have an external parametrisation, but they do not feel any `extra structure' of the reduced manifold, which is nothing more than a smooth manifold.\par
In the `odd directions' we have
\begin{equation}
\theta^{\alpha} \circ \gamma_{\mathbb{R}^{0|p}}(t) =  \sum_{n\ \textnormal{odd}} \frac{1}{n!}  \zx^{i_{1}}\zx^{i_{2}} \cdots \zx^{i_{n}}\theta^{\alpha}_{i_{n} i_{n-1}  \cdots i_{1}}(t)\, .
\end{equation}
  However, the components of $\gamma$ in the `odd directions' do not have any interpretation as `standard curves'  on a purely odd supermanifold. Heuristically, the above shows the need for odd parameters when dealing with curves on a supermanifold.  %
\begin{remark}
Several other notions of curves have appeared in the literature. For example, a supercurve is often understood as being in $\Hom(\mathbb{R}^{1|1},M)$ (see, e.g., \cite{Garnier2012,Goertsches2008,OngayLarios:1992}) or $\InHom(\mathbb{R}^{1|1},M)$ (see, e.g., \cite{Dumitrescu2008}). However, as we do not wish to try to make sense of an odd component of time we reject these notions as suitable for quasi--classical mechanics. An argument \emph{for} the notion of `super--time' can be found in \cite{Salgado:2009}, together with many references of earlier related works. Of course, we do not reject `super-space methods' as a powerful tool in constructing supersymmetric theories, but nonetheless in mechanics, time is the only physically meaningful parameter.  Moreover, in \cite{Bruce:2014} it was shown that $S$--curves are necessary for a  kinematic  definition of the (total spaces of the) tangent bundle and the $k^\textrm{th}$ order tangent bundle of a supermanifold. 
\end{remark}
%
%
We stress that defining a  curve as a functor is vital for our understanding of differential equations and the notion of dynamics.  Our philosophy is that one can never really understand a curve as  being parametrised by a single chosen supermanifold (or Grassmann algebra, see Remark \ref{rem21}), but rather we use \emph{all} supermanifolds as parametrisations at the `same time': this is reflected in the functorality of the constructions (see also Example \ref{exExplicitDE} and Proposition \ref{propFunctoriality} later on).
\subsection{Superdiffeomorphism groups} The idea of the \emph{superdiffeomorphism group} of a supermanifold $M$, which we will denote as $\InDiff(M)$, is intuitively clear (see, e.g., \cite{Sachse2011}). We restrict our attention to the subfunctors of $\InHom(M,M)$ that consist of maps that are invertible with respect to the composition \eqref{eq5Star}--\eqref{eq5Star5Star}. To be more clear on this, we first need the unit element.
\begin{definition}
The \emph{unit element} of $\InHom(M,M)$ consists of the subfunctor given by
\begin{eqnarray*}
\textnormal{id}: \catname{SM}^{\rmo}  &\longrightarrow &\catname{Set}\, ,\\
  S &\longmapsto&  \textnormal{id}_{S} := \Id_{M} \circ \textnormal{prj}_{M}\, .
\end{eqnarray*}
\end{definition}
Here $\textnormal{prj}_{M} : S \times M \rightarrow  M$ is the  obvious projection.  It is a straightforward exercise to show that $\textnormal{id}$ has the expected properties under the composition. 
\begin{definition}
A map $\Phi:  M\longrightarrow M$ is said to be \emph{invertible} if there exists another map $\Phi^{-1}$ such that
\begin{equation}\label{eqHeart}
(\Phi\underline{\circ} \Phi^{-1})_{S} = (\Phi^{-1}\underline{\circ} \Phi)_{S} = \textnormal{id}_{S}\, ,\end{equation}
for all parametrisations $S \in \catname{SM}$.
\end{definition}

\noindent There is no need for the inverse of a map to exists, but when it does it can be shown to be unique.  Informally we will also refer to $\Phi_{S}$ as \emph{invertible} remembering that invertibility is in the sense of \eqref{eqHeart}.\par
We are now ready to define what we mean by a superdiffeomorphism group.
\begin{definition}
The \emph{superdiffeomorphism group} of a supermanifold $M$, which we denote as $\InDiff(M)$, consists of all the subfunctors of $\InHom(M,M)$ that are invertible with respect to the composition $\underline{\circ}$ defined by \eqref{eq5Star}. In other words,
\begin{equation}\label{eqDefSupDiffGroup}
 \InDiff(M)(S) := \left\{\Phi_{S} \in \InHom(M,M)(S)~ | ~~ \Phi_{S}~ \textnormal{is invertible}  \right\}\, .
\end{equation}
\end{definition}
It is a simple exercise to see that we can (and should) consider the superdiffeomorphism group of a supermanifold as a functor 
$$\InDiff(M): \catname{SM}^{\rmo} \longrightarrow \catname{Grp}\, ,$$
where we denote the category of set--theoretical groups as $\catname{Grp}$. Changes of the parametrising supermanifolds lead to group homomorphisms, as they should. \par
Note that we do not spell out the locally ringed space structure of superdiffeomorphism groups, we will only need to consider them as generalised supermanifolds.  Indeed, the most clear notion of a super Lie group is as a functor from the (opposite) category of supermanifolds to groups. \par
\begin{example}
 Superdiffeomorphism groups for purely odd supermanifolds are representable, and so are genuine Lie supergroups. For instance,   from Example \ref{exAntTangBund} it is clear that $\InHom(\mathbb{R}^{0|1},  \mathbb{R}^{0|1}) =  \Pi \sT \mathbb{R}^{0|1} \simeq \mathbb{R}^{1|1}$.   Then in order to constuct $\InDiff(\mathbb{R}^{0|1})$ we restrict to invertiable maps, which amounts to `removing the point zero' from $\mathbb{R}^{1|1}$. Thus $\InDiff(\mathbb{R}^{0|1}) \simeq \mathbb{R}^{1|1}_{\star} \subset \R^{1|1}$. Natural coordinates on $\InDiff(\mathbb{R}^{0|1})$ can be inhereted from $\mathbb{R}^{1|1}$ in the obvious way.  The group law, defined at the level of $S$--points is given by
\begin{equation}\label{eqExGroup}
(x_{S}, \theta_{S}) \cdot (x_{S}^{\prime},  \theta_{S}^{\prime}) = (x_{S} \: x_{S}^{\prime} , \theta_{S} + \theta_{S}^{\prime})\, ,
\end{equation}
\noindent where the multiplication and addition are understood as being in the algebra of functions on $S$. We leave it as an exercise to the reader to show that \eqref{eqExGroup} defines  indeed   a group structure.
\end{example} 
The \emph{Lie superalgebra} of  $\InDiff(M)$ is given by the Lie superalgebra of vector fields $(\Vect(M), [~,~])$, where the bracket is the standard graded Lie bracket of vector fields. As standard, one interprets the action of vector fields on tensors and tensor--like objects on a supermanifold via the Lie derivative as an infinitesimal action of the superdiffeomorphism group. Note that the difference with the  case  of classical manifolds  is that we have both even and odd vector fields, which are understood as derivations of the algebra of functions on the supermanifold in question. At the level of local coordinates on a supermanifold, given any (homogeneous) vector field $X \in \Vect(M)$ we have an associated (local) infinitesimal superdiffeomorphism
\begin{equation*}
x^{a} \longmapsto x^{a} + \epsilon X^{a}(x)\, ,
\end{equation*}
 where $\epsilon$ is an infinitesimal parameter of Grassmann parity $\widetilde{\epsilon} = \widetilde{X}$. As we have both even and odd vector fields to contend with, we are actually dealing with `supersymmetries'. \par

\subsection{Ordinary differential equations} We define a (first--order) \emph{ordinary differential equation on a supermanifold} $M$ as a subfunctor of the tangent bundle $\sT M$ (see Appendix \ref{AppTM})
\begin{equation}\label{eqDEfImpliDE}
D(S) \subset \sT M(S)\, .
\end{equation}
As it stands, our general understanding of an ordinary differential equation is as an \emph{implicit differential equation}. Indeed, we only   have a substructure of the tangent bundle. \par
Given any \emph{even} vector field  $X \in \Vect(M)$, we can consider it as a morphism of supermanifolds
$$X : M \longrightarrow \sT M,$$
 such that $\pi \circ X = \Id_{M}$, where    $\pi: \sT M \longrightarrow M$ is the  natural projection.   
 \begin{definition}\label{defExpliDE}
An \emph{explicit  ordinary differential equation} is the functor $D: \catname{SM}^{\rmo} \longrightarrow \catname{Set}$ given by
$$D(S) :=  \left\{ n \in \sT M(S)~ ~ | ~ ~\exists ~ m \in M(S) ~ ~ \textnormal{such that}~ n = X\circ m\right\}$$
and
\begin{eqnarray}
\psi^{D}: D(S) &\longrightarrow& D(P)\, ,\label{eqDefPsiD}\\
 n & \longmapsto& n' =  n \circ \psi\, ,\nonumber
\end{eqnarray}
  where $\psi$ is the change of parametrisation \eqref{eqChangParam}.
\end{definition}

\begin{remark}\label{rem4Janusz}
 As we are dealing with the (total space of the) tangent bundle of a supermanifold, we are discussing Grassmann even differential equations only. Moreover, it is clear that we cannot consider a Grassmann odd vector field or a inhomogeneous vector field as defining a morphism (in the category of supermanifolds) from $M$ to $\sT M$. Odd vector fields can be considered as morphism $M \rightarrow \Pi \sT M$, where $\Pi$ is the parity reversion functor (see also Example \ref{exAntTangBund}). Thus, a parallel theory of Grassmann odd differential equations can be constructed, though doing so is outside the scope of this paper.
\end{remark}

We then take an \emph{implicit ordinary differential equation} to be an ordinary differential equation that is \emph{not} an explicit ordinary differential equation. 
\begin{definition}\label{defSolutDE}
A \emph{solution to an ordinary differential equation} $D$, is an $S$--curve on $M$ such that
$$\st \gamma_{S}(t) \in D(S)\, ,$$
\noindent for all $t \in \mathbb{R}$.
\end{definition}
In the above definition, we have used the  obvious extension of the standard tangent prolongation of a curve; take the derivative with respect to time and remember that functions on any $S$ are constant (see \eqref{eqSTangentProlong}). 

\begin{example}[An implicit differential equation] Let $M = \R^{1|1}$, which we equip with global coordinates $(x, \theta)$. Then consider the sub--supermanifold  $ D \subset\sT \R^{1|1} \simeq \R^{2|2}$, which we equip with adapted coordinates $(x, \theta, \dot{x}, \dot{\theta})$, defined by the equation
\begin{equation*}
\dot{x} \:  - \: \theta \dot{\theta} = 0\, .
\end{equation*}
Thus, in our understanding we have an implicit ordinary differential equation: just on dimensional grounds we see that this equation cannot be explicit. Solution to this equation are $S$--curves $\gamma \in \InHom(\R, \R^{1|1})$ such that
\begin{equation}\label{eqEX1}
\frac{\rmd}{\rmd t} (x \circ \gamma_{S}(t)) - (\theta \circ \gamma_{S}(t)) \; \frac{\rmd}{ \rmd t}(\theta \circ \gamma_{S}(t)) = 0\, .
\end{equation}
Note that \eqref{eqEX1} cannot be simplified the standard way we are used to, i.e., by replacing $\theta \dot{\theta}$ by $\frac{1}{2}\frac{\rmd}{\rmd t} (\theta^2)$ along $\gamma_S$, due to the anticommuting nature of the Grassmann odd coordinates.
\end{example}
\begin{example}[An explicit differential equation]\label{exExplicitDE}
 Let $M:=\R^{0|2}$, which we provide  with global coordinates $(\theta_+,\theta_-)$, and consider the even vector field
  \begin{equation}\label{eqVFX}
X:=\theta_-\frac{\partial}{\partial\theta_+}-\theta_+\frac{\partial}{\partial\theta_-}\in \Vect(M)\, .
\end{equation}
Our differential equation is the one defining the trajectories of $X$, viz.
\begin{eqnarray}
\frac{d}{dt}(\theta_+\circ\gamma_S(t))&=&\theta_-\circ\gamma_S(t)\, ,\label{eqIntCurvX1}\\
\frac{d}{dt}(\theta_-\circ\gamma_S(t))&=&-\theta_+\circ\gamma_S(t)\, ,\label{eqIntCurvX2}
\end{eqnarray}
where $\gamma_S$ is the unknown $S$--curve. The system \eqref{eqIntCurvX1}--\eqref{eqIntCurvX2} can be solved via a direct integration, yielding to the global solution $\gamma_S$, unambiguously defined by the pull--backs of the global coordinates on $M$. That is,
\begin{eqnarray}
\theta_+\circ\gamma_S(t) &=& A_S\cos (t) + B_S\sin(t)\, , \label{eqSolIntCurvX1}\\
\theta_-\circ\gamma_S(t) &=& B_S\cos(t)-A_S\sin(t)\, ,\label{eqSolIntCurvX2}
\end{eqnarray}
with 
\begin{equation}\label{eqInPoints}
A_S=\theta_+\circ\gamma_S(0)\, ,\quad\textrm{and}\quad B_S=\theta_-\circ\gamma_S(0)\, .
\end{equation}
 It is worth observing that \eqref{eqSolIntCurvX1}--\eqref{eqSolIntCurvX2} may be trivial for particular choices of $S$. Specifically, if $S$ is a genuine manifold (for example just a point), then the left--hand sides vanish as they are the pull--backs of odd coordinates on $M$ to $S$.  In other words, $(A_S, B_S)$ is the initial $S$--point
\begin{equation}
\gamma_S(0):S\longrightarrow M\, .
\end{equation}
In the spirit of $S$--points, in order to analyse the solution \eqref{eqSolIntCurvX1}--\eqref{eqSolIntCurvX2}, one needs to `probe' it by picking a supermanifold $S$ and  we can confine ourselves to Grassmann algebras (see Remark \ref{rem21}). Let us illustrate this by choosing, e.g.,  $S=\R^{0|3}$, which is large enough to clarify the general situation.  \par
Fix global coordinates $(\zeta^1,\zeta^2,\zeta^3)$ on $\R^{0|3}$. Then the initial data \eqref{eqInPoints} are just   generic odd functions on $\R^{0|3}$, i.e., 
\begin{eqnarray*}
A_{\R^{0|3}}&=& a_1\zeta^1+a_2\zeta^2+a_3\zeta^3+a_4\zeta^1\zeta^2\zeta^3\, ,\\
B_{\R^{0|3}}&=& b_1\zeta^1+b_2\zeta^2+b_3\zeta^3+b_4\zeta^1\zeta^2\zeta^3\, .
\end{eqnarray*}
Then, by replacing the above data into \eqref{eqIntCurvX1}--\eqref{eqIntCurvX2}, we have
\begin{eqnarray}
\theta_+\circ\gamma_{\R^{0|3}}(t) &=& \zeta^1(a_1\cos(t)+b_1\sin(t))\label{eqThetaGammaR3-1}\\
&&+ \zeta^2(a_2\cos(t)+b_2\sin(t))\nonumber\\
&&+ \zeta^3(a_3\cos(t)+b_3\sin(t))\nonumber\\
&&+ \zeta^1\zeta^2\zeta^3(a_4\cos(t)+b_4\sin(t))\, ,\nonumber\\
\theta_-\circ\gamma_{\R^{0|3}}(t) &=& \zeta^1(b_1\cos(t)-a_1\sin(t))\label{eqThetaGammaR3-2}\\
&&+ \zeta^2(b_2\cos(t)-a_2\sin(t))\nonumber\\
&&+ \zeta^3(b_3\cos(t)-a_3\sin(t))\nonumber\\
&&+ \zeta^1\zeta^2\zeta^3(b_4\cos(t)-a_4\sin(t))\, .\nonumber
\end{eqnarray}
Let us now discuss some changes of parametrisation (see \eqref{eqChangParam}), namely,
\begin{equation}\label{eqDiagChangCoord}
\xymatrix{
\R^{0|3}\ar@{->>}[r]  &  \R^{0|2}\ar@{->>}[r]\ar[d]^{\psi}   &   \R^{0|1}  \ar@{->>}[r]    & \R^{0|0}\\
               &  \R^{1|1}\, , &&
}
\end{equation}
where the horizontal arrows are the natural projections, and $\psi$ is given by
\begin{equation}\label{eqChangeParamPsi}
(s,\zeta)\circ\psi = (\zeta^1\zeta^2,\zeta^1)\, ,
\end{equation}
with $(s,\zeta)$ global chart on the new `probe' $S=\R^{1|1}$.  
According to our functorial definition  \eqref{eqDEfImpliDE} of an implicit differential equation, dual to the diagram \eqref{eqDiagChangCoord} there is
\begin{equation}\label{eqDiagChangCoordStar}
\xymatrix{
D(\R^{0|3}) &  D(\R^{0|2})\ar@{_(->}[l]    &  D( \R^{0|1}) \ar@{_(->}[l]     & D(\R^{0|0})\ar@{_(->}[l] \\
               & D( \R^{1|1} )\ar[u]^{\psi^D}\, . &&
}
\end{equation}
Observe that $S$--points of $D(\R^{0|3}) $ are described by \eqref{eqThetaGammaR3-1}--\eqref{eqThetaGammaR3-2}. The $S$--points of  $D(\R^{0|2})$ are described by the very same equations, but now with $a_3=b_3=a_4=b_4=0$. The situation is for $D(\R^{0|1}) $ and $D(\R^{0|0})$ is clear. So, the horizontal arrows in \eqref{eqDiagChangCoordStar} are the obvious set--theoretical inclusions. \par
Now let us discuss the vertical arrow in \eqref{eqDiagChangCoord} and \eqref{eqDiagChangCoordStar}. With $S={\R^{1|1}}$, 
 the initial data \eqref{eqInPoints} are     generic odd functions on $\R^{1|1}$, that is
\begin{equation*}
A_{\R^{1|1}}=\zeta A(s)\, ,\quad B_{\R^{1|1}}=\zeta B(s)\, .
\end{equation*}
According, the equations \eqref{eqSolIntCurvX1}--\eqref{eqSolIntCurvX2} read
\begin{eqnarray*}
\theta_+\circ\gamma_{\R^{1|1}}(t) &=& \zeta(A(s)\cos (t) + B(s)\sin(t))\, ,\, \label{eqSolIntCurvX1bis}\\
\theta_-\circ\gamma_{\R^{1|1}}(t) &=& \zeta(B(s)\cos(t)-A(s)\sin(t))\, ,\label{eqSolIntCurvX2bis}
\end{eqnarray*}
and these are   $S$--points of $D( \R^{1|1} )$. 
Now we can describe   the map $\psi^D$ appearing in \eqref{eqDiagChangCoordStar}. According to the definition of $\psi^D$ (see \eqref{eqDefPsiD}), we have   that, by Taylor expansion of $A$ and $B$,
\begin{eqnarray*}
\theta_+\circ\gamma_{\R^{1|1}}(t)\circ\psi &=& \zeta^1(A(0)\cos (t) + B(0)\sin(t))\,  , \\
\theta_-\circ\gamma_{\R^{1|1}}(t)\circ\psi &=& \zeta^1(B(0)\cos(t)-A(0)\sin(t))\, , 
\end{eqnarray*}
and so
\begin{equation*}
 \zeta A(s)\circ\psi=\zeta^1 A(0)\, ,\quad \zeta B(s)\circ\psi=\zeta^1 B(0)\, .
\end{equation*}
  This way, we get a solution \eqref{eqSolIntCurvX1}--\eqref{eqSolIntCurvX2} (i.e., when $S=\R^{0|2}$, with global coordinates $(\zeta^1,\zeta^2)$) with $A(0)=a_1$, $B(0)=b_1$, all other coefficients being zero. So, we have explained how $\psi^D$ acts on an $S$--point of $D( \R^{1|1} )$ and gives out an $S$--point of $D( \R^{0|2} )$.
  \end{example}
The above example illustrates the functorial properties of solutions to ordinary differential equations on supermanifolds, this feature is essential in our general understanding of differential equations and their solutions. We will concentrate on Lagrangian mechanics and phase dynamics in the next section.

\section{The Lagrangian and phase dynamics}\label{sec:Lag}
\subsection{The Lagrangian and its evaluation on curves} We understand a (first--order time--independent) Lagrangian
\begin{equation}\label{eqLag}
L : \sT M \longrightarrow \mathbb{R}
\end{equation}
 to be a morphism of supermanifolds, i.e., an \emph{even} function $L \in C^{\infty}(\sT M)$ on the total space $\sT M$ of the tangent bundle of a supermanifold $M$.
 
\begin{example}[see \cite{Casalbuoni:1976,Junker:1995}]
A simple non--trivial  Lagrangian \eqref{eqLag} with  $M = \mathbb{R}^{1|2}$, which we equip with (global) coordinates $(x, \theta_{+}, \theta_{-})$ is
$$L = \frac{1}{2} \dot{x}^{2} -V(x) + \frac{1}{2}(\dot{\theta}_{+} \theta_{+} - \dot{\theta}_{-}\theta_{-}) - W(x)\theta_{+}\theta_{-}\, ,$$
where $U(x)$ and $V(x)$ are smooth functions in the even variable only.
\end{example}
 
In order to evaluate   \eqref{eqLag}   along an $S$--curve
\begin{equation}\label{eqSCurvLagr}
\gamma_{S} \in \InHom(\mathbb{R}, M)(S)\, ,
\end{equation}
we first use   tangent prolongation
\begin{equation}\label{eqSCurvProlong}
\st \gamma_{S} \in \InHom(\mathbb{R},\sT M)(S),
\end{equation}
of \eqref{eqSCurvLagr},
which is defined similarly to the classical case by
\begin{equation}\label{eqSTangentProlong}
(x^{a}, \dot{x}^{b})\circ \st \gamma_{S}(t) =  \left(x^{a}\circ \gamma_{S}(t),~ \dot{x}^{b}\circ \frac{\rmd \gamma_{S}}{\rmd t}(t)  \right)\, ,
\end{equation}
where $(x^{a}, \dot{x}^{b})$ are the standard local coordinates on $\sT M$ induced from \eqref{eqLocCoord}.   
The curve \eqref{eqSCurvProlong} then `tracks out' $S$--points of $\sT M (S)$ as we vary $t \in \mathbb{R}$, while the curve \eqref{eqSCurvLagr} describes their projections on $M(S)$.\par
Now we can explain how  the evaluation of \eqref{eqLag} along \eqref{eqSCurvLagr} works: just think of the former as a   natural transformation (hopefully without too much of a notational clash)
\begin{equation}
L : \InHom(\mathbb{R},M) \longrightarrow \InHom(\mathbb{R}, \mathbb{R})\, ,
\end{equation}
 viz.
\begin{equation*}
\xymatrix{
\InHom(\mathbb{R},M)(S)\ar[r]^{\overline{\psi}} \ar[d]^{L_S}& \InHom(\mathbb{R},M)(P)\ar[d]^{L_P}\\
\InHom(\mathbb{R},\mathbb{R})(S)\ar[r]_{ {\psi^L}} &\InHom(\mathbb{R},\mathbb{R})(P)\, ,
}
\end{equation*}
where   $\psi^{L}$ is defined by
\begin{equation*}
L \circ \st \gamma_{S} \longmapsto L \circ \st \gamma_{P} := (L \circ \st \gamma_{S} )\circ (\psi, \Id_{\mathbb{R}})\, ,
\end{equation*}
and $\psi$ is the change of parametrisation \eqref{eqChangParam}. 
Then the evaluation of $L$ on $\gamma$ consists in applying $L_S$ to $\gamma_S$, which in fact yields to a parametrised  family  
\begin{equation*}
L_{S}(\gamma(t)) := L \circ \st \gamma_{S}(t) : S \times \mathbb{R} \rightarrow \mathbb{R}\, 
\end{equation*}
 of Lagrangians. In short,  we must always think in terms of families.
 \begin{remark}\label{remREAL}
In physics one insists that the Lagrangians be real. However, this notion is dependent on the reality condition one uses, and in particular the definition of complex conjugation of odd variables. Thus, one may need to  include the complex unit in the local expressions. However, when dealing with examples, we will consider Lagrangians that are manifestly real.
\end{remark}
\subsection{The phase dynamics} Given the Lagrangian \eqref{eqLag}, we can define a morphism 
\begin{equation*}
d L : \sT M \longrightarrow \sT^{*}\sT M\, ,
\end{equation*}
which in local coordinates is given by 
\begin{equation}\label{eqEqDiffLagr}
p_{a}\circ dL = \frac{\partial L}{\partial \dot{x}^{a}} \, ,\quad  \dot{p}_{b} \circ dL =  \frac{\partial L}{\partial x^{b}}\, .
\end{equation}
In complete agreement with the classical case \cite{Tulczyjew:1977}, we have a diffeomorphism of double (super) vector bundles
\begin{equation}\label{eqDiffAlpha}
\alpha : \sT \sT^{*}M \longrightarrow \sT^{*} \sT M\, .
\end{equation}
In the adapted coordinates 
$(x^{a}, \dot{x}^{b}, \dot{p}_{c}, p_{d})$ on $\sT^{*}\sT M$  and 
\begin{equation}\label{eqLocIndCoord}
(x^{a}, p_{b}, \dot{x}_{c}, \dot{p}_{d})
\end{equation}
 on $\sT \sT^{*} M$, induced from \eqref{eqLocCoord}, we can rewrite \eqref{eqDiffAlpha} as
\begin{equation}\label{eqDiffAlphaLOC}
\alpha(x,p, \dot{x}, \dot{p}) = (x, \dot{x}, \dot{p}, p)\, .
\end{equation}
The diffeomorphism \eqref{eqDiffAlpha}, whose local expression \eqref{eqDiffAlphaLOC} does not suffer from any additional factor with respect to the well--known classical case (see, e.g., \cite{Voronov2002}, Theorem 7.1 and/or Appendix \ref{AppDVBs}), is what is needed to define the   \emph{Tulczyjew differential}  
\begin{equation}\label{eqTulczDiff}
\mathcal{T}L := \alpha^{-1} \circ d L\, .
\end{equation}
In local coordinates, \eqref{eqTulczDiff}   amounts to
\begin{equation}\label{eqTulczDiffLOC}
(x^{a}, p_{b}, \dot{x}^{c}, \dot{p}_{d})\circ \mathcal{T}L = \left(x^{a}, \frac{\partial L}{\partial \dot{x}^{b}}, \dot{x}^{c}, \frac{\partial L}{\partial x^{d}}\right)\, ,
\end{equation}
and it fits into the   following diagram
\begin{equation}
\xymatrix{
{\sT \sT^{*}M}\ar[r]^\alpha &   {\sT^{*} \sT M}\ar[d]_\tau\\
&{\sT M}\ar[ul]^{\mathcal{T}L}\, .\ar@/_1.0pc/[u]_{dL}
}
\end{equation}

\begin{definition}\label{defPhDyn}
The \emph{phase dynamics} $\mathcal{D}$ associated with the Lagrangian \eqref{eqLag}, is a functor   
\begin{equation}\label{eqPhDyn}
\mathcal{D} : \catname{SM}^{\rmo} \longrightarrow \catname{Set}
\end{equation}
defined by
\begin{equation}\label{eqDefDS}
\mathcal{D}(S) :=  \mathcal{T}L(\sT M(S)) \subset \sT \sT^{*}M(S)\, 
\end{equation}
and
\begin{eqnarray*}
 \psi^{\mathcal{T}L}: \mathcal{D}(S) &\longrightarrow &\mathcal{D}(P)\, ,\\
  \mathcal{T}L \circ m& \longmapsto& (\mathcal{T}L \circ m) \circ \psi\, ,
\end{eqnarray*}
where $\psi$ is the change of parametrisation \eqref{eqChangParam}.
\end{definition}
It is clear from local description \eqref{eqTulczDiffLOC} that the phase dynamics, as defined by \eqref{eqDefDS}, is actually a sub--supermanifold of $ \sT \sT^{*}M$, that is the phase dynamics is a representable functor. However, the phase dynamics is best described in terms  of the functor of points \eqref{eqSTAR}--\eqref{eqSTARSTAR}, especially when it comes to solutions \eqref{defSolPhDyn}.
\begin{remark}
 The phase dynamics   is in general an \emph{implicit} differential equation (see \eqref{eqDEfImpliDE}): $\mathcal{D} \subset \sT \sT^{*}M$  is  typically  not the `image' of an even vector field on $\sT^{*}M$. However,   we do obtain \emph{explicit} dynamics   when the Lagrangian  \eqref{eqLag} is regular (in the standard meaning). Moreover, if the Lagrangian is regular, then the phase dynamics projects to all of $\sT^{*}M$. For degenerate Lagrangians the phase dynamics usually projects to a sub-supermanifold of $\sT^{*}M$, which corresponds to the sub-supermanifold defined by primary constraints in the sense of Dirac \& Bergmann.  This  in complete agreement with the classical geometry study of differential equations and mechanics (see, e.g., \cite{Marmo:1992}).  
\end{remark}
\begin{remark}
Note that the foundation of our  formalism is the diffeomorphism \eqref{eqDiffAlpha}, rather than   the (even) symplectic structure on $\sT^{*}\sT M$, but this is just a matter of taste, since  the two things are equivalent.  Although we are discussing phase dynamics, Possion brackets are not used here.
\end{remark} 
\subsection{Phase trajectories} By a \emph{Lagrangian system} we mean a pair $(M,L)$, where $M$ is a supermanifold and $L$ is a Lagrangian as in \eqref{eqLag}. 
 The \emph{phase trajectories} of a Lagrangian system $(M,L)$ are the solutions (in the sense of Definition \ref{defSolutDE}) of the phase dynamics \eqref{eqPhDyn} assciated with $L$. 
 \begin{definition}\label{defSolPhDyn}
A \emph{solution of the phase dynamics} \eqref{eqPhDyn} is an $S$--curve 
\begin{equation}\label{eqSolPhDyn}
c_{S} \in \InHom(\mathbb{R}, \sT^{*}M)(S)\,,
\end{equation}
whose tangent prolongation $\st c_{S}$ takes values in $\mathcal{D}(S) \subset \sT \sT^{*}M(S)$ for all $t \in \mathbb{R}$.
\end{definition}
\begin{proposition}\label{propFunctoriality}
If \eqref{eqSolPhDyn} is a solution of the phase dynamics \eqref{eqPhDyn} at $S$, then
$$c_{P} := c_{S} \circ (\psi \circ \Id_{\mathbb{R}}),$$
where $\psi$ is the change of parametrisation \eqref{eqChangParam}, is again a solution of the phase dynamics \eqref{eqPhDyn} at $P$.
\end{proposition}
\begin{proof}
For any $t \in \mathbb{R}$ we have $\st c_{P}(t) = \st c_{S}(t)\circ \psi \in \sT \sT^{*}M(P)$. By assumption $c_{S}$ is a solution to the phase dynamics $\mathcal{D}(S)$ and thus for any and all $t$,  $\st c_{p}(t)$ takes values in $\mathcal{D}(P)$.
\end{proof}
Proposition \ref{propFunctoriality} states that once we have a solution of the phase dynamics in one parametrisations then we can change parametrisations and still have a solution of the phase dynamics. 
\begin{definition}
The \emph{solution space} of the phase dynamics $\mathcal{D}$ is the generalised supermanifold
$$\textnormal{Sol}_{\mathcal{D}} : \catname{SM}^{\rmo} \longrightarrow \catname{Set}$$
defined by
$$\textnormal{Sol}_{\mathcal{D}}(S) := \{ c_{S} \in \InHom(\mathbb{R}, \sT^{*}M)(S) ~ |~ \st c_{S}(t) \in \mathcal{D}(S)~ \textnormal{for all} \: t \in \mathbb{R} \}.$$
\end{definition}
\subsection{The Euler--Lagrange equations}  Let   $\gamma_{S} \in \InHom(\mathbb{R},M)(S)$   denote the $S$--curve underlying the $S$--curve \eqref{eqSolPhDyn}. That is, $\gamma_{S}  :=  \pi \circ c_{S}$, where $\pi : \sT^{*}M \rightarrow M$ is the obvious vector bundle projection. Then, in the coordinates \eqref{eqLocIndCoord}  solutions to  the phase dynamics (Definition \ref{defSolPhDyn}) are (locally) described by 
\begin{eqnarray}\label{eqPEQ}
 x^{a}_{S}(t)& :=& x^{a}\circ \gamma_{S}(t)\, ,\\
\nonumber  p_{b}^{S}(t)& := &p_{b}\circ c_{S}(t) = \left(\frac{\partial L}{\partial \dot{x}^{b}} \right)\circ \st \gamma_{S}(t)\, ,\\
 \nonumber \dot{x}_{S}^{c}(t) &:= & \dot{x}^{c}\circ \st \gamma_{S}(t) = \frac{\rmd x^{c}_{S}}{\rmd t}(t)\, ,\\
 \nonumber \dot{p}_{d}^{S}(t)&:= &\dot{p}_{d}\circ \st c_{S}(t) = \left(\frac{\partial L}{\partial x^{d}} \right)\circ \st \gamma_{S}(t) = \frac{\rmd p_{d}^{S}}{\rmd t}(t)\,  ,
\end{eqnarray}

 \noindent which we refer to as the \emph{phase equations}. Where hence we arrive at the expected \emph{Euler--Lagrange equations} (which we refer to as the equations of motion)
\begin{equation}\label{eqEL}
\frac{\rmd }{\rmd t }\left(\left(\frac{\partial L}{\partial \dot{x}^{a}}\right)\circ \st \gamma_{S}(t)\right) {-} \left(\frac{\partial L}{\partial x^{a}} \right)\circ \st \gamma_{S}(t) =0\, .
\end{equation}
Notice that \eqref{eqEL} is  not really   a single Euler--Lagrange equation, but rather a family of them, in accordance to our functorial definition \eqref{eqDEfImpliDE} of differential equations. According to    Remark \ref{rem21}, to solve \eqref{eqEL} one may pick an arbitrary  Grassmann algebra with a sufficiently large   number of generators and understand any constants of integration (which maybe even and odd) as taking values in this Grassmann algebra. Proposition \ref{propFunctoriality} implies that we do not have spell--out our parametrisation, thus confirming  our general   attitude that nothing physically meaningful can be derived from a fixed parametrisation (recall also Remark \ref{rem21}).   
\subsection{Constants of motion}   We say that a function $f \in C^{\infty}(\sT^{*}M)$ is a \emph{constant of motion} for the Lagrangian system $(M,L)$ if and only if
$$f \circ c_{S}(t) \in \InHom(\mathbb{R}, \mathbb{R}^{1|1})(S)$$
is constant (independent of time $t \in \mathbb{R}$) for all $c_{S} \in \textnormal{Sol}_{\mathcal{D}}(S)$.  
Under the change of parametrisation \eqref{eqChangParam} we see that
$c_{S} \longmapsto c_{P}= c_{S}\circ (\psi, \Id_{R})$ 
implies that 
$$\frac{\rmd}{\rmd t}(f\circ c_{P}(t)) = \frac{\rmd}{\rmd t}(f\circ c_{S}(t)) \circ \psi =0\, .$$
Thus under $\bar{\psi} : \InHom(\mathbb{R}, \sT^{*}M)(S) \longrightarrow \InHom(\mathbb{R}, \sT^{*}M)(P)$  the definition of a constant of motion is unaffected.
\subsection{Symmetries of the phase dynamics} As we are dealing with dynamics on supermanifolds it is natural that we consider both even and odd transformations. That is we must consider the superdiffeomorphism groups $\InDiff $ defined by \eqref{eqDefSupDiffGroup} as opposed to just the diffeomorphism groups (i.e., invertible morphisms in the category of supermanifolds). \par
Let $\Phi \in \InDiff(\sT^{*}M)$ be a  superdiffeomorphism, let $\mathcal{D}$ be the phase dynamics \eqref{eqPhDyn}  associated with the Lagrangian  \eqref{eqLag}, and denote by $c_{S} \in \textnormal{Sol}_{\mathcal{D}}(S)$   a solution  of the phase dynamics $\mathcal{D}(S)$ (recall Definition \ref{defSolPhDyn}).  In the spirit of  \cite{Marmo:1992}, we give the following definition. 
\begin{definition}\label{defSymPhDyn}
$\Phi  $ is a \emph{symmetry of the phase dynamics}   $\mathcal{D}$ if and only if    
$(\Phi \underline{\circ}c)_{S} \in \textnormal{Sol}_{\mathcal{D}}(S)$   for all $S \in \catname{SM}$.  
\end{definition}
By   \emph{symmetries of the Lagrangian system} $(M,L)$ we simply mean the symmetries of the corresponding phase dynamics  $\mathcal{D}$, according to Definition \ref{defSymPhDyn} above,  
that is   subfunctors of $\InDiff(\sT^{*}M)$   mapping solutions of  $\mathcal{D}$ to themselves. Observe that the collection of all symmetries of a Lagrangian system form a super Lie group, under the composition in $\InDiff(\sT^{*}M)$. Then we have   a functor from the (opposite) category of supermanifolds to set--theoretical groups, whose representability is an interesting question, though not discussed here. 
\subsection{Infinitesimal symmetries of the phase dynamics} Again we will follow the ethos of \cite{Marmo:1992} making only the necessary changes for our formalism.
\begin{definition}
 We will say that a function $F \in C^{\infty}(\sT \sT^{*}M)$ \emph{vanishes on the phase dynamics} if and only if
$F\circ n =0$
  for all $S$--points $n \in \mathcal{D}(S)$.
 \end{definition}
Functions vanishing on the phase dynamics clearly form an ideal of $C^{\infty}(\sT \sT^{*}M)$. For any homogeneous vector field $X \in \Vect(\sT^{*}M)$, denote by 
\begin{equation}\label{eqTgLift}
\mathcal{L}_{X} := \rmd_{\sT}X
\end{equation}
 its tangent lift  (see \cite{MR1482898}, Section 4, and also \cite{MR0350650}). Let $(M,L)$ be a Lagrangian system and $ \mathcal{D}$ be the corresponding phase dynamics.
\begin{definition}\label{defInfSymmPhDyn}
$X$ is an \emph{infinitesimal symmetry} of $ \mathcal{D}$    if and only if
\begin{equation}\label{eqDefSymmPhDyn}
(\mathcal{L}_{X}F)\circ n=0\, ,
\end{equation}
  for all   $n \in \mathcal{D}(S)$ and all functions $F$ vanishing on the phase dynamics. \end{definition}
Definition \ref{defInfSymmPhDyn} immediately extends by linearity to  non--homogeneous vector fields. Plainly, infinitesimal symmetries form a sub Lie superalgebra of the Lie superalgebra of vector fields on $\sT^{*}M$ equipped with the standard graded Lie bracket.   \par
In terms of solutions to the phase dynamics (see Definition \ref{defSolPhDyn}), the condition \eqref{eqDefSymmPhDyn} above is clearly equivalent to
$$(\mathcal{L}_{X}F)\circ \st c_{S}(t)=0\, ,$$
 for all time $t \in \mathbb{R}$.
In particular, in the adapted local coordinates \eqref{eqLocIndCoord}, the functions
\begin{equation*}
 \phi_{a} := p_{a}- \frac{\partial L}{\partial \dot{x}^{a}}\, ,\quad  \hat{\phi}_{b} := \dot{p}_{b} - \frac{\partial L}{\partial x^{a}}\, ,
\end{equation*}
obtained from \eqref{eqEqDiffLagr}, both vanish on the phase dynamics. Moreover, it is convenient  to consider the ideal of functions that vanish on the phase dynamics to be generated locally by the set of  functions $\{\phi_{a}, \hat{\phi}_{b} \}$. This is clear from the local description \eqref{eqTulczDiffLOC} of the image of the Tulczjyew differential \eqref{eqTulczDiff}.  
\begin{proposition}\label{propLocForm}
A homogeneous vector field $X$ is an infinitesimal symmetry of $\mathcal{D}$ if and only of if 
\begin{equation*}
 (\mathcal{L}_{X}\phi_{a})\circ \st c_{S}(t)=0\, ,\quad \textnormal{and}\quad (\mathcal{L}_{X}\hat{\phi}_{b})\circ \st c_{S}(t)\, ,
\end{equation*}
\noindent for any and all $c_{S} \in \textnormal{Sol}_{\mathcal{D}}(S)$ for all time $t \in \mathbb{R}$.
\end{proposition}

We can interpret infinitesimal symmetries as infinitesimal diffeomorphisms of the tangent bundle of the phase space that preserve the phase dynamics, and  arise from the tangent lifts  of vector fields on the phase space. Proposition \ref{propLocForm} allows to deduce the local form of infinitesimal symmetries.

 \section{Examples of supermechanical systems}\label{sec:examples}
 
 In this section we will present a few simple examples of supermechanical systems in order to highlight the basic constructions.  We will consider only Lagrangians that are inherently real, one may need to insert factors of the complex unit to match the Lagrangians presented in the physics literature (recall Remark \ref{remREAL}). We will focus primarily on constructing the phase dynamics and presenting solutions of the Euler--Lagrange equations. We will not be considering particularly complicated systems and so we can directly integrate the equations of motion.  A full study of specific models is outside the scope of the paper.

\begin{example}[1--dimensional Dirac--like Lagrangian]
As a very simple, but non--trivial, example we consider the supermechanical theory that is derived by a dimensional reduction of the $(1+1)$--dimensional Dirac Lagrangian in light--cone coordinates (absorbing factors of the complex unit).  Here we have $M = \mathbb{R}^{0|2}$ and the Lagrangian \eqref{eqLag} is
$$L = \frac{1}{2}(\psi_{+}\dot{\psi}_{+} + \psi_{-}\dot{\psi}_{-}) - m \psi_{+}\psi_{-}\, .$$
Here $m$ is a real parameter that we interpret as the quasi--classical mass of the particle. Note that this Dirac--like Lagrangian is degenerate, as is typical for theories with physical fermions. The (global) coordinate description of the Tulczyjew differential \eqref{eqTulczDiffLOC} read now 
\begin{align*}
& \pi_{+}\circ \mathcal{T}L = - \frac{1}{2}\psi_{+}\,  , &&  \pi_{-}\circ \mathcal{T}L  = - \frac{1}{2}\psi_{-} \,  ,\\
& \dot{\pi}_{+}\circ \mathcal{T}L =  \frac{1}{2}\dot{\psi}_{+}-  m \psi_{-}\,  , &&  \dot{\pi}_{-}\circ \mathcal{T}L  =  \frac{1}{2}\dot{\psi}_{-} +  m \psi_{+}\,  .
\end{align*} 
The phase equations \eqref{eqPEQ} boil down to the following
\begin{align*}
&  - \frac{1}{2} \psi_{+}\circ \gamma_{S}(t) = \pi_{+} \circ c_{S}(t)\,  , && - \frac{1}{2}\psi_{-}\circ \gamma_{S}(t) = \pi_{-} \circ c_{S}(t)\,  , \\
& \frac{\rmd}{\rmd t}(\pi_{+}\circ c_{S}(t)) =  \frac{1}{2}\frac{\rmd}{\rmd t}(\psi_{+}\circ \gamma_{S}(t))-  m\:  \psi_{-}\circ \gamma_{S}(t)\,  , &&  \frac{\rmd}{\rmd t}(\pi_{-} \circ c_{S}(t))  =  \frac{1}{2}\frac{\rmd}{\rmd t}(\psi_{-}\circ \gamma_{S}(t))+  m \: \psi_{+}\circ \gamma_{S}(t)\,  .
\end{align*} 
\noindent The Euler--Lagrange equations \eqref{eqEL} can be written out fully explicitly as
\begin{equation}\label{eqDiracELEQ}
 \frac{\rmd}{\rmd t}(\psi_{+}\circ \gamma_{S}(t)) -  m \psi_{-}\circ  \gamma_{S}(t)=0\, , \quad \frac{\rmd}{\rmd t}(\psi_{-}\circ \gamma_{S}(t))+ m \psi_{+}\circ \gamma_{S}(t)  =0\, ,
\end{equation}
 as rather expected. Moreover we can solve the phase dynamics quite explicitly, however we only need to present solutions to \eqref{eqDiracELEQ}, as extending this to the phase dynamics is straightforward. The reader will have notice the similarity between   \eqref{eqDiracELEQ}, and the equations  discussed in Example \ref{exExplicitDE}. In particular, we know how to integrate them: 
\begin{align*}
& \psi_{+} \circ \gamma_{S}(t) =  A_{S} \cos(mt) + B_{S}\sin(mt)\,  ,\\
& \psi_{-} \circ \gamma_{S}(t) =  B_{S} \cos(mt) -  A_{S}\sin(mt)\,  ,
\end{align*}
\noindent where $\psi_{+} \circ \gamma_{S}(0) =  A_{S}$ and $\psi_{-} \circ \gamma_{S}(0) =  B_{S}$ are the initial $S$--points, cf. \eqref{eqInPoints}.\par

Now we pass to the question of symmetries. Rather than presenting the most general infinitesimal symmetry, we will only give one example. By using Proposition \ref{propLocForm}, we claim that  the even vector field
$$X =  \psi_{-} \frac{\partial}{\partial \psi_{+}} - \psi_{+}\frac{\partial}{\partial \psi_{-}} + \pi_{-} \frac{\partial}{\partial \pi_{+}} - \pi_{+} \frac{\partial}{\partial \pi_{-}}$$
 is an infinitesimal symmetry of the phase dynamics. We leave details of the calculation to prove our claim to the reader. We only comment that this infinitesimal symmetry comes from the well known fermionic rotation symmetry $\delta \psi_{+} = \epsilon \psi_{-}$,  $\delta \psi_{-} = -\epsilon \psi_{+}$.\par
  In the standard formalism  the associated Noether charge is just $R =  \psi_{+}\psi_{-}$. Thus, we recover the well known fact that the interaction term between $\psi_{+}$ and $\psi_{-}$ in this and similar models is independent of time, provided the equations of motion are applied. The interested reader can easily check that the general solution to the Euler--Lagrange equations \eqref{eqDiracELEQ} satisfies this condition.
\end{example}

 \begin{example}[$N=2$ supersymmetric mechanics \cite{Witten:1981}] The configuration supermanifold here is  $M = \mathbb{R}^{1|2}$, which we will  equip with global coordinates $(x, \psi_{+}, \psi_{-})$. The Lagrangian \eqref{eqLag} for this model is given by
 $$L = \frac{1}{2}\dot{x}^{2} + \frac{1}{2}U^{2}(x) + \frac{1}{2}\left(\dot{\psi}_{+} \psi_{+} - \dot{\psi}_{-}\psi_{-} \right) +  U^{\prime}(x)\psi_{+}\psi_{-}\, .$$
This Lagrangian describes the supersymmetric mechanics of a particle in one (even) dimension in a potential $-U^{2}(x)$.   The (global) coordinate expression of the Tulczyjew differential \eqref{eqTulczDiffLOC}  now read
 \begin{align*}
 & p \circ \mathcal{T}L = \dot{x}\, , && \pi_{+}\circ \mathcal{T}L = \frac{1}{2} \psi_{+}\, , && \pi_{-}\circ \mathcal{T}L = - \frac{1}{2} \psi_{-}\, ,\\
 & \dot{p} \circ \mathcal{T}L =  U(x)U^{\prime}(x) + U^{\prime \prime} \psi_{+}\psi_{-}\, ,&& \dot{\pi}_{+}\circ \mathcal{T}L = -\frac{1}{2} \dot{\psi}_{+} +  U^{\prime}(x)\psi_{-}\, , && \dot{\pi}_{-}\circ \mathcal{T}L = \frac{1}{2} \dot{\psi}_{-} -  U^{\prime}(x)\psi_{+}\, .
 \end{align*}
 The phase equations \eqref{eqPEQ} then boils down to the following equations on $S$--curves
 \begin{align*}
 & \frac{\rmd}{\rmd t} (x\circ \gamma_{S}(t))= p \circ c_{S}(t)\, , && \frac{1}{2}\psi_{+}\circ \gamma_{S}(t) = \pi_{+}\circ c_{S}(t)\, , \\
 & - \frac{1}{2}\psi_{-}\circ \gamma_{S}(t) = \pi_{-}\circ c_{S}(t)\, ,&& 
  \frac{\rmd}{\rmd t}(p \circ c_{S}) =  ( U(x)U^{\prime}(x) +  U^{\prime \prime}(x)\psi_{+}\psi_{-})\circ \gamma_{S}(t)\, ,\\
  & \frac{\rmd}{\rmd t}(\pi_{+} \circ c_{S}(t)) =  -\frac{1}{2} \frac{\rmd}{\rmd t}(\psi_{+}\circ \gamma_{S}) +  (U^{\prime}(x)\psi_{-})\circ  \gamma_{S}(t)\, , && \frac{\rmd}{\rmd t}(\pi_{-} \circ c_{S}(t)) =  \frac{1}{2}\frac{\rmd}{\rmd t}(\psi_{-} \circ \gamma_{S}(t)) - (U^{\prime}(x)\psi_{+})\circ  \gamma_{S}(t)\, .
 \end{align*}
Observe that the  Euler--Lagrange equations \eqref{eqEL} we obtain here, viz.
 \begin{eqnarray}
 \frac{\rmd^{2}}{\rmd t^{2}}(x \circ\gamma_{S}(t)) &=& (U(x)U^{\prime}(x) +  U^{\prime \prime}(x)\psi_{+}\psi_{-})\circ \gamma_{S}(t)\, ,  \label{eqEX2ELEQ1} \\
 \frac{\rmd }{\rmd t }(\psi_{+} \circ \gamma_{S}(t)) &=& (U^{\prime}(x)\psi_{-}) \circ \gamma_{S}(t)\, ,\label{eqEX2ELEQ2}\\  \frac{\rmd }{\rmd t }(\psi_{-} \circ \gamma_{S}(t)) &=&  (U^{\prime}(x)\psi_{+}) \circ \gamma_{S}(t),\label{eqEX2ELEQ3}
 \end{eqnarray}
depend on the choice of the potential $U(x)$. 
 As the intention of this work is not to get heavily involved in methods of finding solutions to the phase dynamics or Euler--Lagrange equations,\footnote{ For solutions of the Euler--Lagrange equations for this and similar models (i.e. with different reality conditions chosen) one can consult \cite{Heumann:2000,Junker:1996,Manton:1999}. Note that the proceeding works require that  all the dynamics variables take values in some chosen Grassmann algebra. } for illustration purposes we pick the \emph{harmonic potential} $U(x) = k x$, where $k$  is taken to be a real number. With this potential the equations \eqref{eqEX2ELEQ1}--\eqref{eqEX2ELEQ2}--\eqref{eqEX2ELEQ3} for the even and odd degrees of freedom decouple and can be solved via direct integration. Explicitly\footnote{The reason of the $1/k$ factor in front of $b_S$ is that it makes $b_S$ coincides with the derivative of $x$ at zero.} we have
 \begin{align*}
 & x \circ \gamma_{S}(t) = a_{S} \cosh(k t) + \frac{b_{S}}{k} \sinh(k t)\, ,\\
 & \psi_{+}\circ \gamma_{S}(t) =  A_{S} \cosh(k t) + B_{S}\sinh(kt)\, ,\\
 & \psi_{-}\circ \gamma_{S}(t) =  B_{S} \cosh(k t) + A_{S}\sinh(kt)\, ,
 \end{align*}
 \noindent where the (`functor valued') integration constants are defined by our initial conditions
 \begin{align*}
 & x \circ \gamma_{S}(0) = a_{S}\, , && \frac{\rmd }{\rmd t}\left.(x\circ \gamma(t))\right |_{t=0} = b_{S}\, ,\\
 & \psi_{+}\circ \gamma_{S}(t) = A_{S}\, , && \psi_{-}\circ \gamma_{S}(t) = B_{S}\, .
 \end{align*}
 Returning to the case of a general potential $U(x)$, this model exhibits two supersymmetries, which in `physics notation' are usually written as
 \begin{align*}
 & \delta_{1}x = \epsilon \psi_{+}\, , && \delta_{1}\psi_{+}= \epsilon \dot{x}\, , && \delta_{1}\psi_{-} = \epsilon U(x)\, ,\\
 & \delta_{2}x = \bar{\epsilon} \psi_{-}\, , && \delta_{2}\psi_{+}= -\epsilon U(x)\, , && \delta_{2}\psi_{-} = - \epsilon \dot{x}\, ,
 \end{align*}
 \noindent where $\epsilon$ and $\bar{\epsilon}$ are independent Grassmann odd parameters. Note that, as we present it, these supersymmetries are \emph{on--shell}, that is we have a symmetry provided we impose the equations of motion \eqref{eqEX2ELEQ1}--\eqref{eqEX2ELEQ2}--\eqref{eqEX2ELEQ3}.  By using Proposition \ref{propLocForm}, we claim that the corresponding  vector fields describing these supersymmetries, now understood as symmetries of the phase dynamics, are
 \begin{eqnarray*}
 X_{1} &=& \psi_{+}\frac{\partial}{\partial x} + p \frac{\partial}{\partial \psi_{+}} + U(x)\frac{\partial }{\partial \psi_{-}} + U^{\prime}(x)\psi_{-}\frac{\partial}{\partial p} + \frac{1}{2}p \frac{\partial}{\partial \pi_{+}} - \frac{1}{2} U(x)\frac{\partial}{\partial \pi_{-}}\, ,\\
 X_{2} &=& \psi_{-}\frac{\partial}{\partial x} - U(x)\frac{\partial}{\partial \psi_{+}} - p\frac{\partial }{\partial \psi_{-}} + U^{\prime}(x)\psi_{+}\frac{\partial}{\partial p} - \frac{1}{2}U(x) \frac{\partial}{\partial \pi_{+}} + \frac{1}{2} p\frac{\partial}{\partial \pi_{-}}\, .
 \end{eqnarray*}
 \noindent It is left as an exercise for the reader to show that by tangent lifting (see \eqref{eqTgLift}) and using the equations of motion \eqref{eqEX2ELEQ1}--\eqref{eqEX2ELEQ2}--\eqref{eqEX2ELEQ3} that these vector fields are indeed symmetries of the phase dynamics.\par
 \end{example}

  \begin{example}[Non--holonomic constraints]Rather than attempting to develop a full theory of \emph{constrained supermechanical systems} we present one simple example. We start with a rather general Lagrangian on $M = \mathbb{R}^{1|2}$, which we equip with (global) coordinates $(x, \psi_{+}, \psi_{-})$, 
 \begin{equation}\label{eqnLagrangianConstrained}
 L = \frac{1}{2}\dot{x}^{2} + \frac{1}{2}\left(\dot{\psi}_{+} \psi_{+} - \dot{\psi}_{-} \psi_{-} \right) -U(x)\psi_{+}\psi_{-}\, .
 \end{equation}
 We then want to implement the \emph{linear} non--holonomic constraint $\dot{\psi}_{-} =0$. As standard, if we were to calculate the Euler--Lagrange equations for this system and then enforce the constraint we end up with the wrong equations of motion. To remedy this, we need a constrained version of the Tulczyjew differential. Let us denote the vector bundle that describes the constraint by $E$. Then we have the natural embedding
 $$\iota : E \hookrightarrow \sT \mathbb{R}^{0|2},$$
 given in local coordinates as
 $$(x, \psi_{+}, \psi_{-} , \dot{x}, \dot{\psi}_{+}, \dot{\psi}_{-}) \circ \iota = (x, \psi_{+}, \psi_{-} , \dot{x}, \dot{\psi}_{+}, 0)\,  .$$
The dual vector bundle $E^{*}$ we naturally equip with coordinates $(x, \psi_{+}, \psi_{-}, p, \pi_{+})$. We have a dual morphism, which in this case is just the obvious  projection 
$$\iota^{\dag} : \sT^{*} \mathbb{R}^{1|2}  \rightarrow E^{*}.$$
The \emph{constrained Tulczyjew differential} $ \mathcal{T}L^{E} : E \longrightarrow \sT E^{*}$ is then defined as
$$ \mathcal{T}L^{E} = \sT \iota^{\dag} \circ \alpha^{-1} \circ dL\, .$$
We need not change a word in the definition of the phase dynamics or solutions thereof, (Definition \ref{defPhDyn} \& Definition \ref{defSolPhDyn}), other than replacing the standard Tulczyjew differential  with the constrained Tulczyjew differential (cf. \cite{Grabowski:2009}), viz.
\begin{align*}
& p \circ \mathcal{T}L^{E} =  \dot{x}\, , & & \pi_{+} \circ \mathcal{T}L^{E} = \frac{1}{2} \psi_{+}\, ,\\
& \dot{p} \circ \mathcal{T}L^{E} = - U^{\prime} (x) \psi_{+}\psi_{-}\,  , && \dot{\pi}_{+} \circ \mathcal{T}L^{E} = - \frac{1}{2}\dot{\psi}_{+}- U(x)\psi_{-}\, .
\end{align*}
In particular the phase equations \eqref{eqPEQ} can easily be deduced and the associated Euler--Lagrange equations \eqref{eqEL} are given by
\begin{align*}
& \frac{\rmd^{2}}{\rmd t^{2}}(x \circ \gamma_{S}(t) )  = - \left(U^{\prime}(x)\psi_{+}\psi_{-} \right ) \circ \gamma_{S}(t)\, ,\\
& \frac{\rmd}{\rmd t}(\psi_{+}\circ \gamma_{S}(t)) = - \left(U(x)\psi_{-}\right) \circ \gamma_{S}(t)\, ,\\
& \frac{\rmd}{\rmd t}(\psi_{-}\circ \gamma_{S}(t)) =0\, .
\end{align*}
These equations of motion are identical (up to the obvious relabelling) to the equations of motion  of  the $N$=1 supersymmetric model (see \cite{Manton:1999} and references therein) given by
$$L^{\prime} = \frac{1}{2} \dot{x}^{2} + \frac{1}{2}\dot{\psi} \psi + \lambda U(x) \psi\, , $$
where $\lambda$ is a Grassmann odd constant. However, we cannot consider this `Lagrangian' as a genuine function on $\sT \mathbb{R}^{1|1}$ due to the odd constant: thus it cannot define (phase) dynamics in our understanding (Definition \ref{defPhDyn}). Our resolution to this is clear, we really have the  Lagrangian \eqref{eqnLagrangianConstrained} on $\mathbb{R}^{1|2}$ subject to a linear non--holonomic constraint. 
 \end{example}

 \begin{example}[Geodesics on the super--sphere $\mathbb{S}^{2|2}$]
 Let us equip   $\R^{3|2}$ with (global) coordinates $(x,y,z,\psi_+,\psi_-)$. Then the equation
  \begin{equation}
x^2+y^2+z^2-  \psi_{+}\psi_{-}=1
\end{equation}
defines  the super--sphere
  \begin{equation}
\mathbb{S}^{2|2}\subset\R^{3|2}\, .
\end{equation}
As (local) coordinates on $M:=\mathbb{S}^{2|2}$ we can use the standard `angles' $(\theta,\phi)$, i.e., essentially the coordinates inherited from using polar coordinates on $\R^{3}$, complemented by the   `odd stuff'  $(\psi_+,\psi_-)$ inherited from the environment. The underlying manifold of $M$ is the standard two--sphere $|M| = \mathbb{S}^2$. As a sub--supermanifold of the super--Riemannian manifold  $\R^{3|2}$, our supermanifold $M$ is equipped with a non--degenerate metric inhered by the embedding, which in the above coordinates reads
\begin{equation}
g=\left(\begin{array}{c|c}g_0 & 0 \\\hline 0 & J_2\end{array}\right)\, ,\quad g_0=\left(\begin{array}{cc}1 & 0 \\0 & \sin^2\theta\end{array}\right)\, ,\quad J_2=\left(\begin{array}{cc}0& -1 \\1 & 0\end{array}\right)\, .
\end{equation}
Then  as  the Lagrangian  \eqref{eqLag}, we may take the `length functional'
\begin{equation}
L(\dot{\theta},\dot{\phi}, \dot{\psi}_+,\dot{\psi}_-):=\frac{1}{2}(\dot{\theta}^2+\sin^2\theta\dot{\phi}^2)- \dot{\psi}_+\dot{\psi}_-\, .
\end{equation}
By the definition \eqref{eqTulczDiff} of the Tulczyjew differential $\mathcal{T}L$, we have
\begin{align*}
& p_\theta\circ \mathcal{T}L = \frac{\partial L}{\partial \dot{\theta}}=\dot{\theta}\, , && 
p_\phi\circ \mathcal{T}L = \frac{\partial L}{\partial \dot{\phi}}=\sin^2\theta\dot{\phi}\, ,\\
 & \pi_+\circ \mathcal{T}L =- \dot{\psi}_- \, , && 
\pi_-\circ \mathcal{T}L =- \dot{\psi}_+\, ,
\end{align*}
 where   $p_\theta, p_\phi, \pi_+, \pi_-$ are the momenta of the coordinates $\theta,\phi,\psi_+,\psi_-$  of $M$, respectively (i.e., those collectively denoted by $p_b$ in \eqref{eqLocIndCoord}). An $S$--curve  $c_S$ is then a solution to our phase dynamics (cf. Definition \ref{defSolPhDyn}) if and only if
 \begin{align*}
& p_\theta\circ c_S = \dot{\theta}\circ \st\gamma_S\, ,&  p_\phi\circ c_S = (\sin(\theta)\dot{\phi})\circ \st\gamma_S\, ,& &
\pi_+\circ c_S = -\dot{\psi}_-\circ \st\gamma_S\, , && 
\pi_-\circ c_S = \dot{\psi}_+\circ \st\gamma_S\, ,\\
& \dot{\theta} \circ \st\gamma_S =\frac{\rmd}{\rmd t}(\theta\circ\gamma_S)\, ,&
 \dot{\phi} \circ \st\gamma_S =\frac{\rmd}{\rmd t}(\phi\circ\gamma_S)\, ,&&
\dot{\psi}_+ \circ \st\gamma_S =\frac{\rmd}{\rmd t}(\psi_+\circ\gamma_S)\, ,&&
 \dot{\psi}_- \circ \st\gamma_S =\frac{\rmd}{\rmd t}(\psi_-\circ\gamma_S)\, ,\\
 & \dot{p}_\theta\circ\st  c_S = (\cos\theta\sin\theta\dot{\phi}^2)\circ \st\gamma_S\, ,&
 \dot{p}_\phi\circ\st c_S =0\, ,&&   \pi_+\circ \st c_S = 0\, ,&&   \pi_-\circ \st c_S = 0\, .
\end{align*}
The Euler--Lagrange equations \eqref{eqEL}   look like
\begin{align*}
& \frac{\rmd^2}{\rmd t^2}(\theta\circ\gamma_S(t)) {-}    (\cos\theta\sin\theta )\circ  \gamma_S(t)\: \frac{\rmd^2}{\rmd t^2}(\phi\circ\gamma_S(t))    =0\, , & 
 \sin\theta \circ  \gamma_S(t) \: \frac{\rmd}{\rmd t}(\phi\circ\gamma_S(t))      =\ell_S\, ,\\
 & \frac{\rmd^2}{\rmd t^2}(\psi_+\circ\gamma_S(t)) =0\, ,& 
  \frac{\rmd^2}{\rmd t^2}(\psi_-\circ\gamma_S(t)) =0\, ,
\end{align*}
where $\ell$ is a `Grassmann even constant', i.e., an even function on $S$ (once $S$ is chosen) and thus independent of $t \in \mathbb{R}$. By analogy with the classical case, we  interpret $\ell_S$  the angular momentum of the $S$--curve.\par
Rather than finding the general solution, let us find a particular one, namely a solution $\gamma_S$ of the form
 \begin{align*}
& \theta\circ\gamma_S(t) =\frac{\pi}{2}\, ,& \phi\circ\gamma_S(t) =\ell_S t\, ,\\
& \psi_+ \circ \gamma_{S}(t) = A_St\, ,&
\psi_- \circ \gamma_{S}(t) = B_St\, .
\end{align*}
 Then the angular momentum $\ell_S$ can be interpreted as an initial datum, viz.
\begin{equation*}
\ell_S=\left.\frac{\rmd}{\rmd t}(\phi\circ\gamma_S(t))\right|_{t=0}\, ,
\end{equation*}
which needs to be specified. We finally observe that the initial $S$--point
  \begin{equation*}
(\pi/2,0,0,0) 
\end{equation*}
of $M=\mathbb{S}^{2|2}$, which is of course a genuine point on $\mathbb{S}^{2}$,
evolves `off'  $\mathbb{S}^2$ into $\mathbb{S}^{2|2}$, i.e., it adds nontrivial odd components, provided that the odd velocities  $A_S$ and $B_S$
(playing the same r\^ole as the $A_S$'s and $B_S$'s  in Example \ref{exExplicitDE}) are not simultaneously zero (for all $S\in \catname{SM}$).
 \end{example}

\section{Concluding remarks}\label{sec:con}
 In this work we have presented a different prespective on supermechanics in which the dynamical variables do not take their values in a specified Grassmann algebra, but rather, in a loose sense, take their values in all (real) supermanifolds. That is, we consider the basic dynamical objects, the $S$--curves, as time parametrised functors from the (opposite) category of supermanifolds to sets.  By specifing a (finitely generated) Grassmann algebra we recover the more standard understanding of supermechanics as Grassmann--valued mechanics.  The advantage of our approch is that we do not have to make a choice of the Grassmann algebra that our mechanics takes values in. We use our understanding of $S$-curves to generalise the geometric methods of Tulczyjew to supermechanics. In particular  we define the phase dynamics and solutions thereof in an intrinsic geometric setup that is insensitive to the distinction between non--degenerate and degenerate Lagrangians. \par  
The downside of our approach is that one now has to understand  many of the basic objects in a categorical framework as functors.  This adds a further layer of abstractness to the general theory of supermanifolds and supermechanics. However, this  abstractness seems rather unavoidable.  As supermanifolds  represent a class of noncommutative geometries, maybe it is unrealistic to expect that a `simple' and completley satisfactory notion of a curve --- and so dynamics --- can be found.  For example, supercurves understood as elements of $\Hom(\mathbb{R}^{1|1},M)$ are unsuitiable for defining the tangent and higher tangent bundles of a supermanifold in terms of jets:  while $S$--curves can be used for such constructions (see \cite{Bruce:2014}).  \par 
We have not presented the Hamiltonian side of the Tulczyjew triple in this work. We do not envisage any particular complications in developing the Hamiltonian side of the theory, though this should be checked carefully. Our main reason for sticking to the Lagrangian description is simply that in physics, and particularly field theory, one usually starts from a Lagrangian.\par
Supermechanics we certainly view as a toy--model of quasi--classical field theories. To our knowledge, the framework of geometric field theory --- multisymplectic structures and so on --- has not been applied to supersymmetric field theory. In part, we think that this is due to the lack of appreciation of  categorical methods applied to supergeometry within the geometric mechanics community. We hope that this work is a small step in remedying this situation.

\section*{Acknowledgments}\label{sec:Ack}
 We thank Janusz Grabowski and  Georg Junker for their comments on earlier drafts of this work.\par 
The Research of AJB and KG funded  by the  Polish National Science Centre grant under the contract number DEC-2012/06/A/ST1/00256.  The Research of GM has received funding from the European Union's Horizon 2020 research and innovation programme under the Marie Sk\l odowska--Curie grant agreement No 654721 `GEOGRAL'.
\appendix

\section{Tangent and cotangent bundles}\label{AppTM}
 There are --- just as in the classical case ---  several way to define the (total spaces of) tangent and cotangent bundles in the category of supermanifolds. The most direct definition is to consider then as \emph{natural bundles}:  that is given an atlas of a supermanifold we can canonically build an altas for the tangent bundle  and the cotangent bundle. To do this we need to specify coordinates on these bundles and the induced changes of coordinates.  With this in mind, recall that we can write changes of local coordinates as
$$x^{a} \mapsto x^{a'} = x^{a'}(x),$$
recalling that we require changes of local coordinates to respect the Grassmann parity. Moreover, any change of coordinates is polynomial in the odd coordinates. We can then define the tangent bundle of a supermanifold, which we denote as $\sT M$, as the supermanifold equipped with (adapted) local coordinates $(x^{a}, \: \dot{x}^{b})$, where $\widetilde{\dot{x}^b} = \widetilde{b}$ (i.e., the `dotted' coordinate has the same Grassmann parity as the associated `undotted' coordinate). The induced changes of local (fibre) coordinates --- and so the supermanifold structure ---  of the tangent bundle is given by
$$\dot{x}^{a} \mapsto \dot{x}^{a'} = \dot{x}^{b} \left(\frac{\partial x^{a'}}{\partial x^{b}}\right),$$
which is in complete agreement with the classical case: the coordinates $\dot{x}$ are informally referred to as \emph{velocities}. Similarly, we can define the cotangent bundle of a supermanifold, which we denote as $\sT^{*}M$ as the supermanifold equipped with (adapted) local coordinates $(x^{a},\: p_{b})$, where $\widetilde{p_{b}} = \widetilde{b}$.  The induced changes of (fibre) coordinates are
$$p_{a} \mapsto p_{a'} =  \left(\frac{\partial x^{b}}{\partial x^{a'}} \right)p_{b},$$
again in complete agreement with the classical case: the coordinates $p$ are informally referred to as \emph{momenta}.\par
\section{The canonical diffeomorphism $\alpha$}\label{AppDVBs}
Here we show that there exists a diffeomorphism $\alpha: \sT \sT^{*}M \rightarrow \sT^{*}\sT M$ for any supermanifold $M$. To our knowledge, this was first proved by Th. Voronov \cite{Voronov2002}, and we include a (sketch of a) proof here in order to keep this paper reasonably self-contained. Our approach is via local coordinates. Remembering that momenta transform like derivatives and velocities like differentials (see Appendix \ref{AppTM}), it is a straightforward exersise to show that $\sT^{*} \sT M$ can be equipped with natural local coordinates $(x^{a}, \dot{x}^{b}, q_{c}, \dot{q}_{d})$ and that changes of coordinates are of the form
\begin{align*}
& x^{a'} = x^{a'}(x), && \dot{x}^{b'} = \dot{x}^{c}\left( \frac{\partial x^{b'}}{\partial x^{c}}\right),\\
& q_{c'} = \left( \frac{\partial x^{a}}{\partial x^{c'}}\right)q_{a}, && \dot{q}_{d'} = \left( \frac{\partial x^{a}}{\partial x^{d'}}\right)\dot{q}_{a} + \dot{x}^{c}\left( \frac{\partial x^{e'}}{\partial x^{c}}\right)\left( \frac{\partial^{2} x^{f}}{\partial x^{e'} \partial x^{d'}}\right)q_{f}.
\end{align*}
Not that there are no extra minus signs as compared with the classical situation. \par 
Similarly $\sT \sT^{*} M$ can be equipped with natural local coordinates $(x^{a}, p_{b}, \dot{x}^{c}, \dot{p}_{d})$, the  change of coordinates  for the $p$ and $\dot{p}$ are
\begin{align*}
& p_{d'} = \left( \frac{\partial x^{a}}{\partial x^{d'}}\right)p_{a} + \dot{x}^{c}\left( \frac{\partial x^{e'}}{\partial x^{c}}\right)\left( \frac{\partial^{2} x^{f}}{\partial x^{e'} \partial x^{d'}}\right)\dot{p}_{f}, && \dot{p}_{c'} = \left( \frac{\partial x^{a}}{\partial x^{c'}}\right)\dot{p}_{a}.
\end{align*}
Via inspection we see that, in these adapted local coordinates we have a diffeomorphism $\alpha: \sT \sT^{*}M \rightarrow \sT^{*}\sT M$ given by
$$(x^{a}, \dot{x}^{b}, q_{c}, \dot{q}_{d}) \circ \alpha  = (x^{a}, \dot{x}^{b}, \dot{p_{c}} , p_{d}).$$
Thus we are free to employ  the natural local coordinates on $\sT \sT^{*} M$ as local coordinates on $\sT^{*} \sT M$ and in doing so we can write 
$$\alpha(x,p, \dot{x}, \dot{p}) = (x, \dot{x}, \dot{p}, p),$$
in complete agreement with the classical case (see \cite{Tulczyjew:1977}).

\bibliographystyle{plainnat}
\bibliography{biblio}


\end{document}